\documentclass{article}
\usepackage{arxiv}
\usepackage{graphicx,url}

\usepackage{algorithm,algpseudocode,amssymb,subcaption,natbib}
\usepackage{amsmath, amsthm, amsfonts}
\newtheorem{proposition}{Proposition}
\newtheorem{corollary}[proposition]{Corollary}
\newtheorem{lemma}[proposition]{Lemma}
\theoremstyle{remark}
\newtheorem{remark}[proposition]{Remark}

\providecommand{\Description}[1]{}

\begin{document}

\title{Homotopy-Aware Multi-Agent Path Planning on Plane}
\author {
    Kazumi Kasaura \\
    OMRON SINIC X Corporation \\
    \texttt{kazumi.kasaura@sinicx.com}
}

\maketitle

\begin{abstract}
We propose an efficient framework using Dynnikov coordinates for homotopy-aware multi-agent path planning in planar domains that may contain obstacles.
We developed a method for generating multiple homotopically distinct solutions for the multi-agent path planning problem in planar domains by combining our framework with revised prioritized planning and proved its completeness under specific assumptions.
Experimentally, we demonstrated that our method is significantly faster than a method without Dynnikov coordinates.
We also confirmed experimentally that homotopy-aware planning contributes to avoiding locally optimal solutions when searching for low-cost trajectories for a swarm of agents in a continuous environment.
\end{abstract}

\section{Introduction}\label{sec:introduction}
Path planning with topological constraints or objectives is an important task, because considering topological characteristics of paths is crucial for various aspects of robot planning and navigation.
One notable example is when seeking globally optimal trajectories.
To find a globally optimal trajectory under a complex objective, (such as considering agent kinematics), a conventional strategy involves planning an initial path on a simple graph (e.g., a grid) and optimizing it locally under the objective function~\citep{rosmann2017kinodynamic}.
However, this can lead to local optima, and it is difficult to know beforehand which path will converge to a globally optimal solution.
To avoid this, multiple paths should be optimized and compared. On the other hand, it is redundant to optimize paths that converge to the same solution.
Here, the key is topological characteristics of the paths since optimization does not alter them.
It has been proposed to generate several topologically distinct paths as initial solutions~\citep{kuderer2014online,rosmann2017integrated}.
Figure~\ref{fig:homotopy_example} shows an example of initial paths with the same start and goal.
Because $p_1$ goes through the area above the obstacle and $p_2$ and $p_3$ go through the area below, $p_1$ and $p_2$ (or $p_3$) have different topological features, while $p_2$ and $p_3$ have the same ones.
These facts will be more formally discussed later in \S~\ref{subsec:homotopy}.
Since $p_1$ and $p_2$ (or $p_3$) must converge to different trajectories after optimization and, the better option depends on the objective function.
Therefore, we want to generate both.
Conversely, since $p_2$ and $p_3$ are topologically identical and likely converge to the same trajectory, we do not generate both.
To do this, planning should be performed with the calculation of the topological characteristics of paths, which we call \textit{topology-aware} path planning.
Topology-awareness in planning also finds relevance in tasks such as cable manipulation~\citep{bhattacharya2015topological,kim2014path}, human-robot interaction~\citep{govindarajan2016human,yi2016homotopy}, and high-level planning with dynamic obstacles~\citep{cao2019dynamic}.
\begin{figure}[t]
    \centering
    \includegraphics[width=0.4\linewidth]{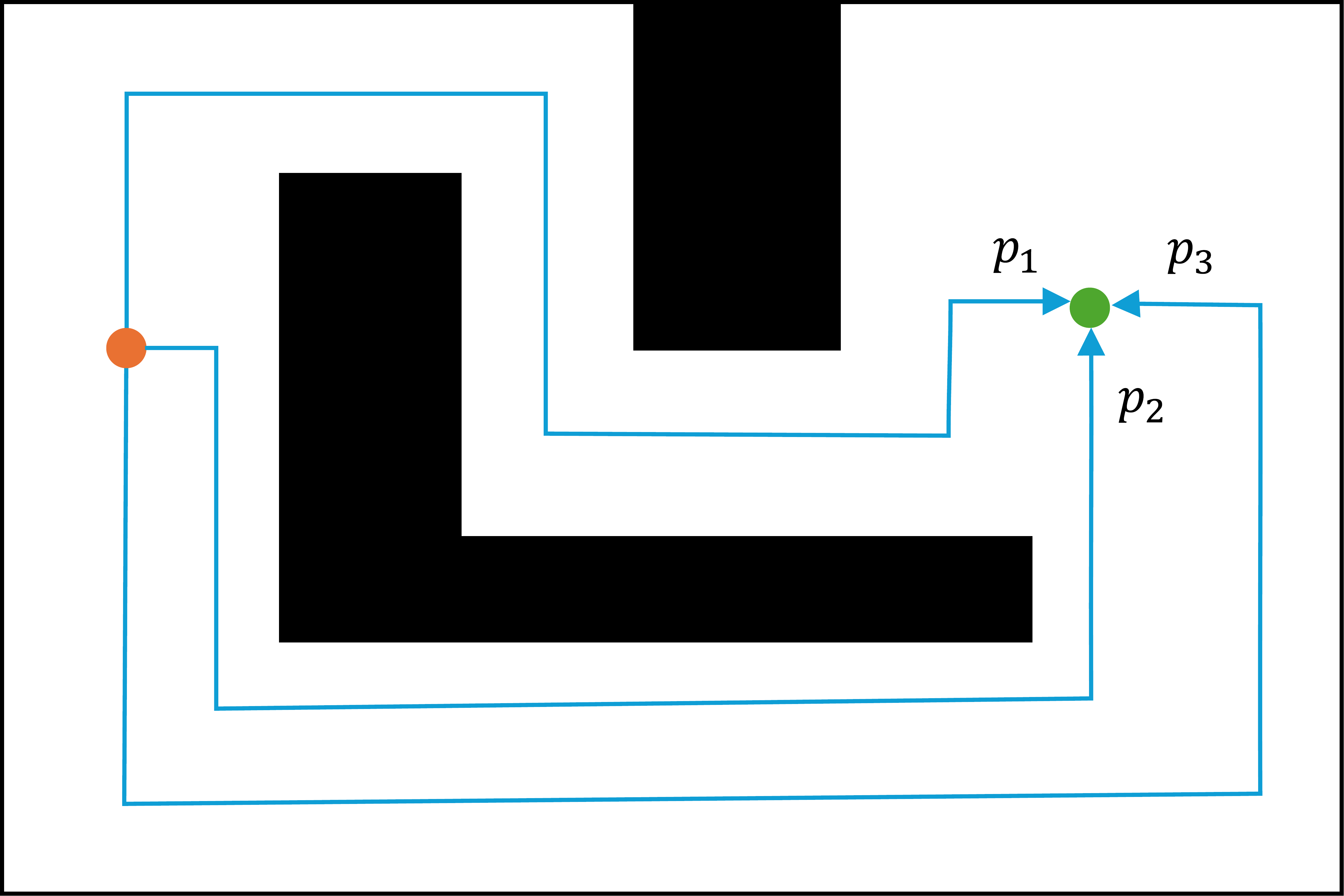}
    \caption{Example of initial coarse paths connecting the same start and goal. After optimizing, $p_1$ and $p_2$ must converge to different smoothed trajectories, while $p_2$ and $p_3$ are likely to converge to the same one.}
    \label{fig:homotopy_example}
    \Description{Three paths in environments with obstacles}
\end{figure}

\textit{Homotopy} is a straightforward topological feature of paths, but difficult to calculate due to its non-abelian nature, as will be discussed later.
\citet{bhattacharya2018path} proposed an approach on the basis of a concept called the {\it homotopy-augmented graph} for general \textit{homotopy-aware} path planning (path planning with the calculation of homotopy classes of paths).
Homotopy-aware path planning using a roadmap is reduced to pathfinding on the homotopy-augmented graph constructed from the roadmap.\footnote{Since we focus on the approach of using roadmaps for path planning, the problem that we treat in this paper is actually pathfinding on graphs. However, the homotopies we consider are not those in discrete graphs~\citep{ghrist1999configuration} but those in the continuous domain. Thus, to avoid confusion, we do not use the term \textit{homotopy-aware pathfinding}.}
To search on a homotopy-augmented graph, we have to manage elements of the \textit{fundamental group} of the searched space.
The fundamental group is not generally abelian and can be computationally difficult to deal with. 
Elements of the group are represented by strings of generators, called \textit{words}.
However, different words can represent the same element of the group, and determining the identity of such representations, termed the {\it word problem}, is not generally solvable~\citep{novikov1955algorithmic}.
Consequently, homotopy-aware path planning remains a generally difficult task.

On the other hand, multi-agent path planning, which plans paths for multiple agents so that they do not collide with each other, has a variety of applications~\citep{silver2005cooperative,dresner2008multiagent,li2020lifelong}.
In multi-agent scenarios on a plane, topological features are nontrivial even in the absence of obstacles since agents are obstacles for other agents. For example, when two agents pass each other, the topological characteristics of the solution vary with the direction of avoidance (counterclockwise or clockwise) as shown in Figure~\ref{fig:two-agents}.
Thus, topological considerations are important for finding the optimal trajectory described in the first paragraph, even in the cases of obstacle-free environments.
Moreover, homotopical considerations are important in the multi-agent case for other reasons.
Homotopy is used for multi-agent coordination~\citep{mavrogiannis2019multi}, especially for avoiding deadlocks~\citep{vcap2016provably}.
However, studies on homotopy-aware planning for multi-agent scenarios, even in obstacle-free environments, are limited.

\begin{figure}[t]
\centering
    \begin{subfigure}{0.3\columnwidth}
        \includegraphics[width=\textwidth]{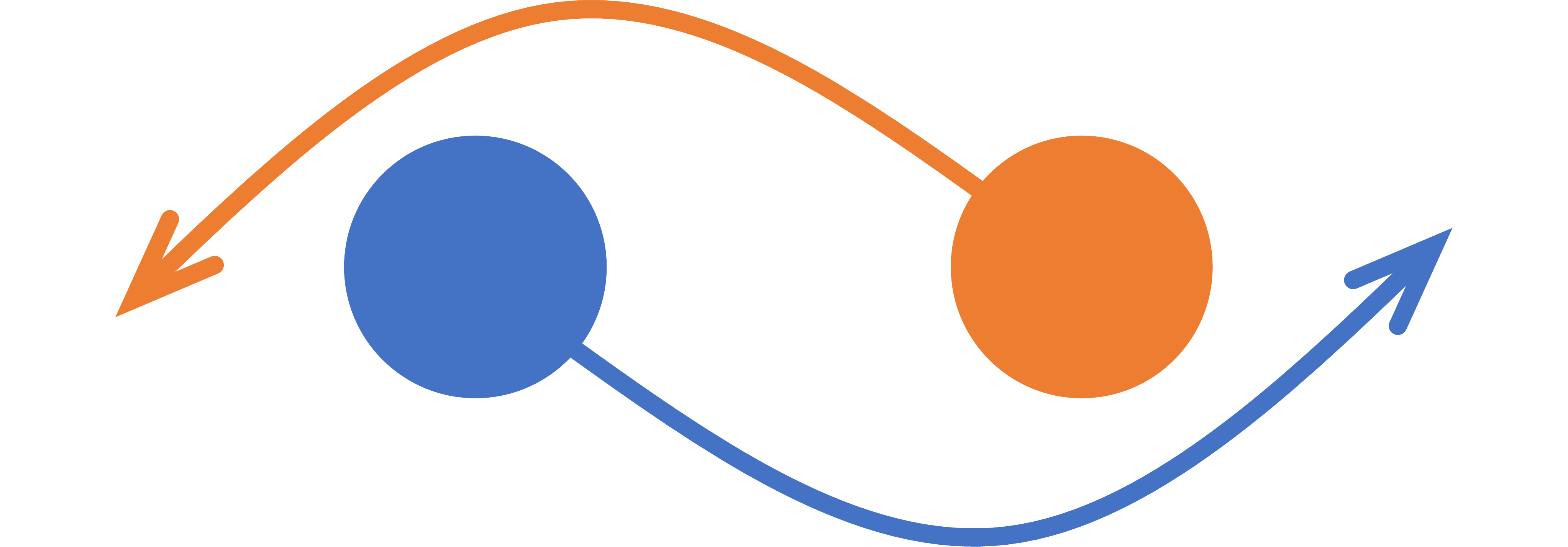}
        \caption{Counterclockwise}
        \label{fig:counterclockwise}
    \end{subfigure}    
    \hspace{10pt}
    \begin{subfigure}{0.3\columnwidth}
        \includegraphics[width=\textwidth]{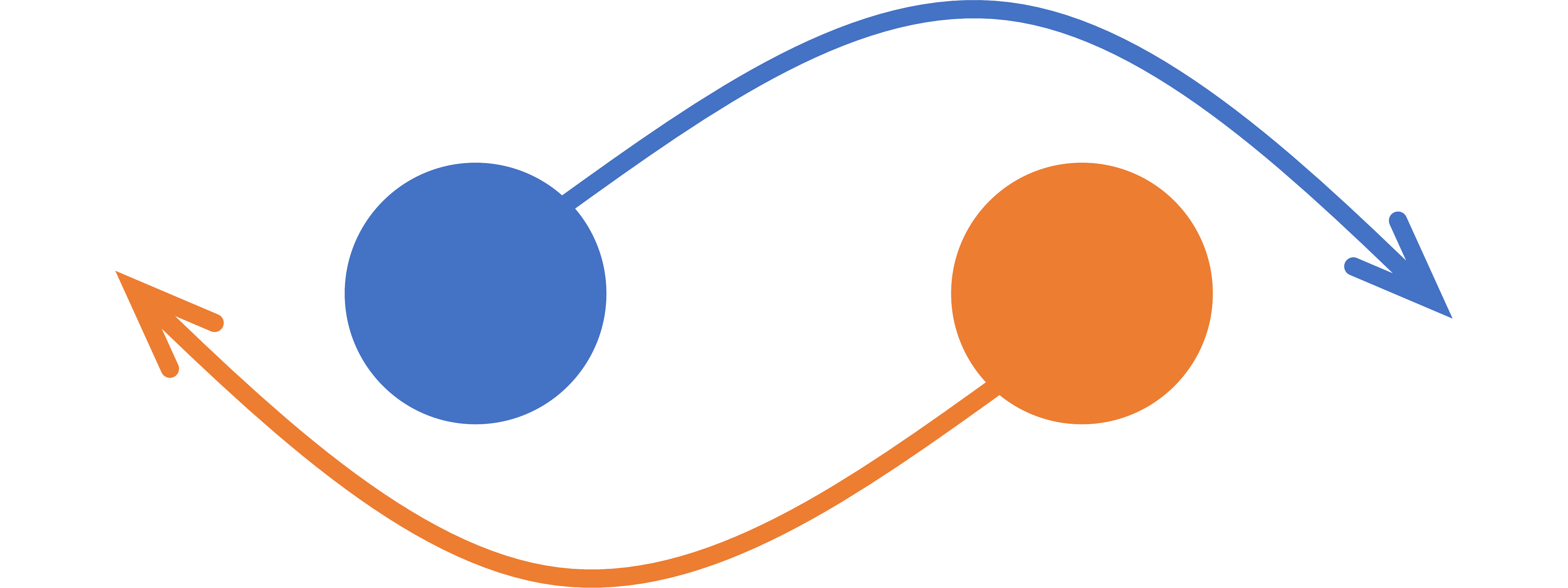}
        \caption{Clockwise}
        \label{fig:clockwise}
    \end{subfigure}
\caption{Two ways for two agents to pass each other. The homotopy class of the path in the configuration space depends on the direction of their passing.}
\label{fig:two-agents}
\Description{Two patterns of two circling arrows, counterclockwise and clockwise}
\end{figure}

By combining the aforementioned concepts with research for braids, we obtained an efficient framework for homotopy-aware multi-agent path planning in planar domains, which is the first valid solution to this problem to our knowledge.
There are three key ideas.
First, while the fundamental group for the multi-agent path planning problem on a plane is the \textit{pure braid group}, we label the homotopy classes of solutions by elements of the \textit{braid group}.
This is equivalent to expanding the space of paths to solutions of \textit{unlabeled} multi-agent path planning, in which, as long as each goal is reached by only one agent, it is permissible for any agent to proceed to any goal.~\citep{adler2015efficient}\footnote{This setting is also called {\it permutation-invariant}~\citep{kloder2006path,yu2013multi} or {\it anonymous}~\citep{stern2019multi}.}.
This simplifies the word construction.
Note that this expansion is only virtual for the calculation of homotopy classes and the actual pathfinding remains labeled.
Since homotopy classes inherently contain information about agent-goal correspondences, this idea does not give rise to any confusion.
The second idea is to use \textit{Dynnikov coordinates}, a representation of elements of the braid group as tuples of integers, which are expected to be calculated at a low computational cost~\citep{dehornoy2008efficient}.
Thanks to these coordinates, we can efficiently maintain a homotopy-augmented graph by using a data structure such as self-balancing binary search trees~\citep{knuth1998art}.
The third is that obstacles in the domain are taken into account by considering them as virtual agents.

We consider the problem of generating multiple homotopically distinct solutions for multi-agent path planning in a planar domain, and provide a method for this problem 
by combining our framework with the {\it revised prioritized planning}~\citep{vcap2015prioritized}.
We also prove that our method can generate solutions belonging to all homotopy classes under certain assumptions.

We experimentally demonstrate that the runtime of our method increases roughly quadratically with respect to the number of agents on a scale of several hundred agents, while the runtime of another method, which uses the Dehornoy order~\citep{dehornoy1994braid} to manage braids instead of Dynnikov coordinates, increases approximately quintically.

In addition, we demonstrated the usefulness of homotopy-awareness for multi-agent trajectory optimization through experimentation.
Specifically, we generated several homotopically different (discrete) solutions using our method. We then continuously optimized them for a certain cost function, and chose the best one. We confirmed that this solution was a globally better trajectory than those generated by baseline methods.

The contributions of this paper are summarized as follows:
\begin{itemize}
\item We propose the first sound framework for homotopy-aware multi-agent path planning on the plane.
\item We provide a method using this framework to generate multiple homotopically distinct solutions for multi-agent path planning in the plane. We have theoretically proved a kind of completeness of our method and experimentally showed its scalability.
\item We experimentally showed that solving this problem contributes to multi-agent trajectory optimization.
\end{itemize}

\section{Related Work}
We survey previous works for multi-agent path planning, topology-aware path planning, and braids.
\subsection{Multi-Agent Path Planning}
Multi-agent pathfinding, a field of study focusing on planning for multiple agents in discrete graphs, has been a subject of extensive research, particularly in grid-based environments~\citep{stern2019multi}.
One of the approaches to this problem is prioritized planning~\citep{erdmann1987multiple,silver2005cooperative}, which is non-optimal, incomplete yet scalable. \citet{vcap2015prioritized} proposed Revised Prioritized Planning (RPP) and proved its completeness under certain assumptions.
Major optimal approaches for multi-agent pathfinding include increasing-cost tree search~\citep{sharon2013increasing}, conflict-based search~\citep{sharon2015conflict}, and some reduction-based methods~\citep{surynek2016efficient,bartak2017modeling}.
Surveys of solutions have been conducted by \citet{stern2019multi2} and \citet{lejeune2021survey}.
While these methods do not consider homotopy, some of them could be modified to do so by combining them with our framework.
In this paper, RPP was adopted for the scalability reason.
On the other hand, we believe that some methods are not suitable for considering homotopy.
See also \S~\ref{sec:conclusion}.

For multi-agent path planning in continuous environments, the typical strategy involves a three-step process: roadmap generation, discrete pathfinding, and continuous smoothing of trajectories. For instance, \citet{honig2018trajectory} presented such an approach for quadrotor swarm navigation.
Since the quality of the smoothed trajectories is difficult to predict beforehand, generating multiple initial solutions, as described in \S~\ref{sec:introduction}, can be effective.
Several roadmap-generation methods tailored for multi-agent scenarios have been proposed~\citep{henkel2020optimized,arias2021avoidance,okumura2022ctrms}. A number of multi-agent pathfinding algorithms have been adapted to handle continuous time scenarios~\citep{yakovlev2017any,walker2018extended,andreychuk2022multi,surynek2019multi}.
Although we used a simple grid for simplicity in our experiments, our method can be extended to treat general roadmaps and continuous time. See Remark~\ref{rem:continous_time}.

\subsection{Topology-Aware Path Planning}
For the single-agent case on a plane, there are various studies, both theoretical and applied, on homotopy-aware path planning.
For planning on a plane with polygonal obstacles,
there exist methods using polygon partition~\citep{park2015homotopy,liu2023homotopy} with analysis of time complexity~\citep{hernandez2015comparison,efrat2006computing,bespamyatnikh2003computing}.
For scenarios involving possibly non-polygonal obstacles, several approaches have been explored~\citep{jenkins1991shortest,hernandez2015comparison,yi2016homotopy,schmitzberger2002capture}.
\citet{grigoriev1997computing,grigoriev1998polytime} proposed a method of constructing words by detecting traversing {\it cuts}, which are also called {\it rays}~\citep{tovar2010sensor}.
Along this idea, the notion of the {\it homotopy-augmented graph} ({\it h-augmented graph}), was proposed and applied to the navigation of a mobile robot with a cable~\citep{bhattacharya2015topological,kim2014path}.

The last approach was generalized by \citet{bhattacharya2018path} to various situations including multi-agent path planning on a plane. However, their algorithm is incomplete.
While they attempted to solve the word problem by using Dehn's algorithm~\citep{lyndon1977combinatorial}, they acknowledged that this algorithm may not always yield accurate results for their presentation.\footnote{Actually, their algorithm fails to perform correctly when dealing with scenarios involving more than three agents (see Appendix~\ref{appendix:pn} for details).}
To our knowledge, there exists no known complete framework for considering homotopy in the multi-agent case on a plane, even when obstacles are absent.

{\it Homology} serves as a coarser classification compared with homotopy, whereby two paths belonging to the same homotopy class are also in the same homology class, but the reverse is not always true\footnote{In the context of robotics, the terms "homotopy" and "homology" are sometimes used interchangeably.\citep{bhattacharya2012topological}}.
While homology is not well-suited for detailed path analysis like homotopy, it possesses the advantage of being computationally easier due to its abelian nature.
Algorithms have been developed for homology-aware path planning in two, three, or higher dimensional Euclidean spaces with obstacles~\citep{bhattacharya2010search,bhattacharya2011identification,bhattacharya2012topological,bhattacharya2013invariants}.
In the planar scenario, homology can be determined by a tuple of winding numbers~\citep{vernaza2012efficiently}.
This number count the number of times that an agent travels around an obstacle.
In the multi-agent planning in the plane, homology can also be determined by a tuple of all winding numbers between agents or between agents and obstacles~\citep{rolfsen2010tutorial}.
These facts have been applied to enable homology-aware planning for mobile robot navigation~\citep{kuderer2014online,mavrogiannis2020multi,mavrogiannis2021hamiltonian}.
The two solutions in Figure~\ref{fig:braid example} in \S~\ref{subsec:word construction} belongs to the same homology class since they have the winding numbers for all agent pairs, while they are homotopically distinct.

\citet{jaillet2008path} introduced the notion of \textit{visibility deformation} (VD), which is stricter than homotopy equivalence: while two paths are homotopic if they are visibility deformable one into the other, the reverse is not always true.
Since it is computationally expensive, \citet{zhou2020robust} introduced a simpler and even stricter notion, \textit{uniform visibility deformation} (UVD): two paths $\gamma, \gamma'$with the same endpoints belong to the same UVD class if, for any $t$, the segment between $\gamma(t)$ and $\gamma'(t)$ is contained within the free space.
While these notions are useful in three-dimensional environments~\citep{zhou2020robust}, they will be too strict in two-dimensional environments for our motivation, because path optimization may not be visibility deformation.
Moreover, these relations are not equivalence relations in general, because, even if the segment between $A$ and $B$ and the segment between $B$ and $C$ are contained within the free space, the segment between $A$ and $C$ may not.

As mentioned in the introduction, several works~\citep{kuderer2014online,rosmann2017integrated,zhou2020robust,de2024topology} focused on the strategy to generate multiple topologically distinct trajectories for avoiding local optima.
For navigation with presence of other agents, it is proposed to enumerate topological patterns of agent coordination and to select the best one, by using winding numbers~\citep{mavrogiannis2020multi,mavrogiannis2021hamiltonian}, which represent homology classes.
It was also proposed to use supervised learning in selecting topology class to imitate human behaviors~\citep{martinez2024shine}.
Our work will enable the extension of these approaches to multi-agent path planning on the plane, without losing the fineness of homotopical classification.

\subsection{Braids and their Applications to Robotics}
The braid group was introduced by \citet{artin1947braids,artin1947theory}.
\citet{fox1962braid} proved that it is the fundamental group of the unlabeled configuration space.
There are several algorithms for solving the word problem of braid groups~\citep{garside1969braid,epstein1992word,birman1998new,hamidi2000complexity,garber2002new}. \citet{dehornoy1994braid} introduced a linear order of braids called the \textit{Dehornoy order}, for which comparison algorithms were presented~\citep{dehornoy1997fast,malyutin2004fast}.
\citet{dehornoy2008ordering} conducted a survey of orders of braids and their comparison algorithms.
It was proved that the braid groups are linear~\citep{bigelow2001braid,krammer2002braid}.
Dynnikov coordinates were introduced by \citet{dynnikov2002yang}.

In robotics, the concept of braid groups on graphs is used for robot planning on graphs~\citep{ghrist1999configuration,kurlin2012computing}.
Regarding the planar case, \citet{diaz2017multirobot} developed a framework enabling the control of agents such that their trajectories correspond to specific braids. In their setting, agents move in circular paths on a predefined track, while our approach deals with path planning with arbitrary start and goal positions.

It was proposed to impose homotopy constraints when executing multi-agent plans to prevent deadlock even in the event of delays~\citep{gregoire2013robust,vcap2016provably}.
Homotopy constraints were also suggested for use in game-theoretic motion planning in urban driving scenarios~\citep{khayyat2024tactical}.
Although these studies did not use the braid group, they could be expanded to address more intricate topological relationships via braids.

Braid groups were used for predicting trajectories of other agents in distributed multi-agent navigation~\citep{mavrogiannis2019multi}. This framework was applied to intersection management~\citep{mavrogiannis2022implicit}.
While we focus on a centralized planning for a relatively large number of agents,
our efforts to compute braids efficiently could benefit research in this area.

\citet{lin2021artin} used braids to represent states of a knitting machine to find its optimal plans. While their aim is quite different from ours, their proposed method is similar to ours since it searches for optimal transfer plans for states represented by braids. While they used symmetric normal forms~\citep{dehornoy2008efficient} to manage braids, we do not for the reason of computational efficiency.

\section{Preliminary}\label{sec:preliminary}
In \S~\ref{subsec:homotopy}, we outline some general notions on topology.
In \S~\ref{subsec:h-graph}, we describe the homotopy-augmented graph, which is a specific notion to homotopy-aware path planning.
In \S~\ref{subsec:braid}, we describe the braid group, which represents homotopy classes for multiple agents on the plane.
To do this, we explain the notion of presentation of a group in \S~\ref{subsec:presentation}.

\subsection{Homotopy and Fundamental Group}\label{subsec:homotopy}
Let $X$ be a topological space.
A continuous morphism from the interval $[0,1]$ to $X$ is called a {\it path}.
Two paths with the same endpoints are said to be \textit{homotopic} if, intuitively, one can be continuously deformed into the other in $X$.
The formal definition is as follows:
Two paths $\gamma_0$ and $\gamma_1$ are homotopic if 
there exists a continuous map $h:[0,1]\times[0,1] \rightarrow X$ such that
\begin{equation}
h(s,0)=\gamma_0(0)=\gamma_1(0),\,h(s,1)=\gamma_0(1)=\gamma_1(1)    
\end{equation}
\begin{equation}
h(0,t)=\gamma_0(t),\,h(1,t)=\gamma_1(t),
\end{equation}
for any $s,t\in [0,1]$. Such map $h$ is called \textit{homotopy}.
The equivalence class of paths under the homotopic relation is called a \textit{homotopy class}. The homotopy class of a path $\gamma$ is denoted as $[\gamma]$.

For example in Fig~\ref{fig:homotopy_example}, paths $p_2$ and $p_3$ are homotopic because we can construct a homotopy between them in the free space. On the other hand, $p_1$ and $p_2$ are not homotopic because the obstacle prevent them from deforming into each other.

For a path $\gamma$, $\gamma^{-1}$ denotes the path that follows $\gamma$ in reverse.
Given two paths $\gamma_0, \gamma_1$ such that $\gamma_0(1)=\gamma_1(0)$, we define the composition $\gamma_0\circ\gamma_1$ as
\begin{equation}
\gamma_0\circ\gamma_1(t)=
\begin{cases}
    \gamma_0(2t) & 0 \leq t \leq 1/2,\\
    \gamma_1(2t-1) & 1/2 \leq t \leq 1.
\end{cases}
\end{equation}

Consider a point, $x_0\in X$. The set of homotopy classes of closed loops with endpoints at $x_0$ (i.e., $\gamma:[0,1]\rightarrow X$ such that $\gamma(0)=\gamma(1)=x_0$) forms a group under path composition. The identity element of this group is the homotopy class of the constant map to $x_0$. The inverse element of $[\gamma]$ is $[\gamma^{-1}]$.
When $X$ is path-connected, this group is independent of the choice of $x_0$ up to isomorphism and denoted as $\pi_1(X)$. It is called the {\it fundamental group of $X$}~\citep{hatcher2002algebraic}.

\begin{figure}
    \centering
    \includegraphics[width=0.3\linewidth]{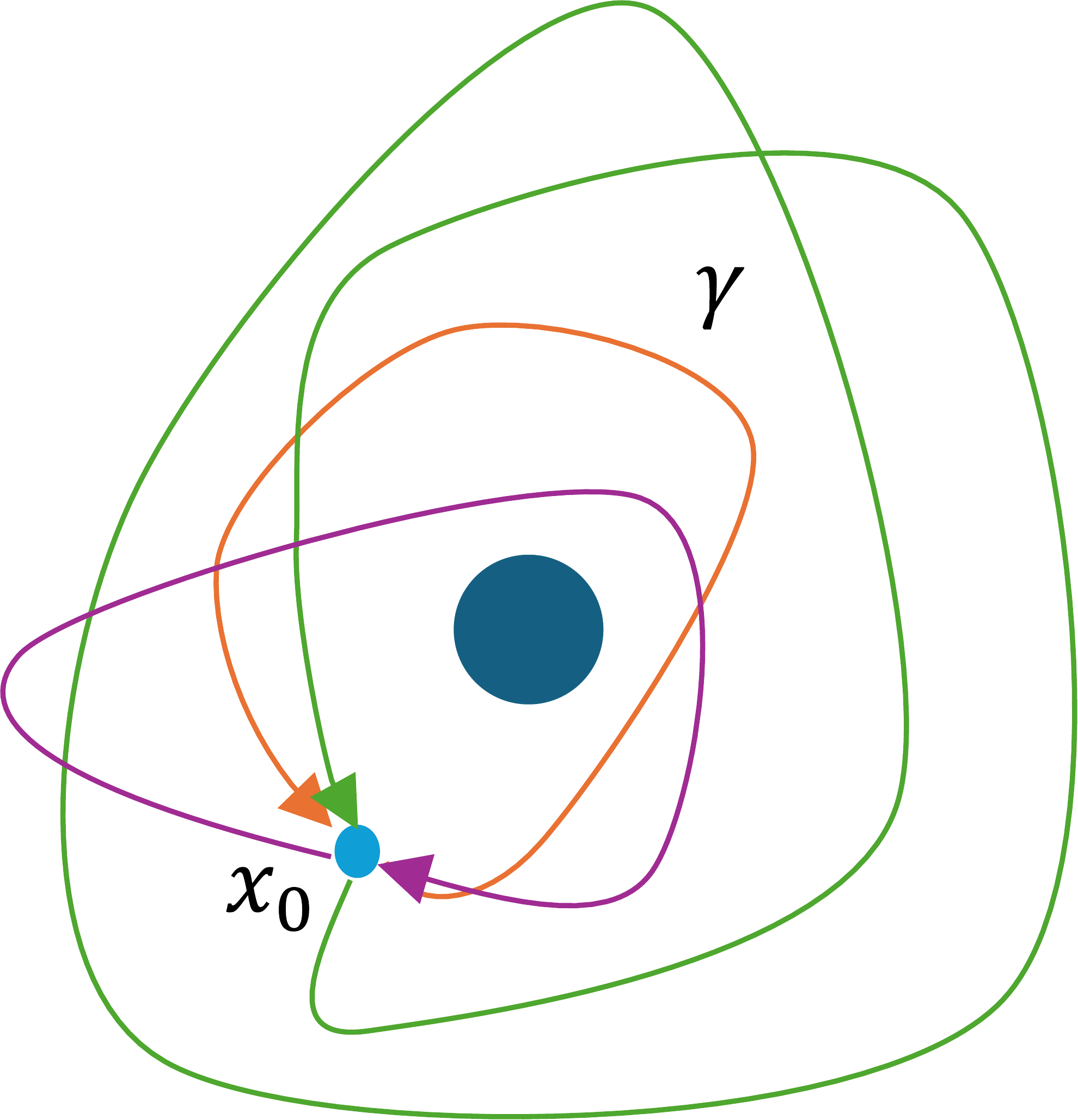}
    \caption{Simplest example for the nontrivial fundamental group where there is a single obstacle (a disk) on the plane. A point $x_0$ is the base point and a loop $\gamma$ from $x_0$ goes around the obstacle counterclockwise once. The loop that goes around the obstacle clockwise once is homotopic to $\gamma^{-1}$. The loop that goes around the obstacle counterclockwise twice is homotopic to $\gamma^{2}$, which traces $\gamma$ twice.}
    \label{fig:fundamental_group}
    \Description{A disk and three circling arrows, counterclockwise, clockwise, and twice counterclockwise}
\end{figure}

For example, when $X=\mathbb{R}^2$, the fundamental group is trivial because all paths with the same endpoints are homotopic.
Figure~\ref{fig:fundamental_group} shows the simplest example for the nontrivial fundamental group where $X=\mathbb{R}^2\setminus D$ and $D$ is a disk. The homotopy class of a loop is determined by how many times it goes around the obstacle and in which direction.
Furthermore, any homotopy class can be written as $[\gamma^n]$ with $n \in \mathbb{Z}$ and the fundamental group is isomorphic to $\mathbb{Z}$

Note that, when $X$ is path-connected, for any two points $x_0, x_1\in X$, there exists a bijection between homotopy classes of paths from $x_0$ to $x_1$ and the fundamental group.
Indeed, when we fix a path $\gamma_0$ from $x_0$ to $x_1$, the homotopy class of any path $\gamma$ from $x_0$ to $x_1$ is represented by $[\gamma_0^{-1}\circ \gamma] \in \pi_1(X)$.

\subsection{Homotopy-Augmented Graph}\label{subsec:h-graph}
Homotopy-aware search-based path planning is reduced to pathfinding on
a {\it homotopy-augmented graph}~\citep{bhattacharya2018path}.
Intuitively, the homotopy-augmented graph is a graph where homotopical information has been added to the vertices.
A path in the original graph can be lifted to a path in the homotopy-augmented graph.
For two paths with the same start and goal points in the original graph, the goal points of the lifted paths are the same if and only if the paths are homotopic.
Therefore, finding multiple non-homotopic paths from a given start to a given goal is reduced to finding paths to multiple vertices of the homotopy-augmented graph that correspond to the goal vertex with different homotopical information.
Mathematically, the homotopy-augmented graph can be considered as a lift of the original graph into the universal covering space. However, we do not define this in this paper.

Let $X$ be a path-connected space, and let $G=(V,E)$ be a discrete graph (roadmap) on $X$.
The homotopy-augmented graph $G_h=(V_h, E_h)$ is constructed as follows\footnote{Our description is generalized from that in the original paper to be independent of how homotopy is calculated.}.
We fix a base point $x_0\in X$ for the fundamental group. We also assume that one homotopy class $[\gamma_{v}]$ of paths from $x_0$ to each vertex $v\in V$ is fixed. The construction of $G_h$ is then expressed as follows:

\begin{itemize}
\item $V_h := V \times \pi_1(X)$, where $\pi_1(X)$ is the fundamental group of $X$ with respect to the chosen base point $x_0$.
\item For each edge $e\in E$ from vertex $v\in V$ to vertex $u \in V$, and for each element $\alpha\in\pi_1(X)$, $E_h$ contains an edge from $(v,\alpha)$ to $(u,\alpha h(e))$, where $h(e):=[\gamma_v\circ e\circ \gamma_u^{-1}]$. Here, $\gamma_u$ and $\gamma_v$ are paths in homotopy classes $[\gamma_u]$ and $[\gamma_v]$, respectively.
\end{itemize}

\begin{figure}[t]
    \centering
    \includegraphics[width=0.40\columnwidth]{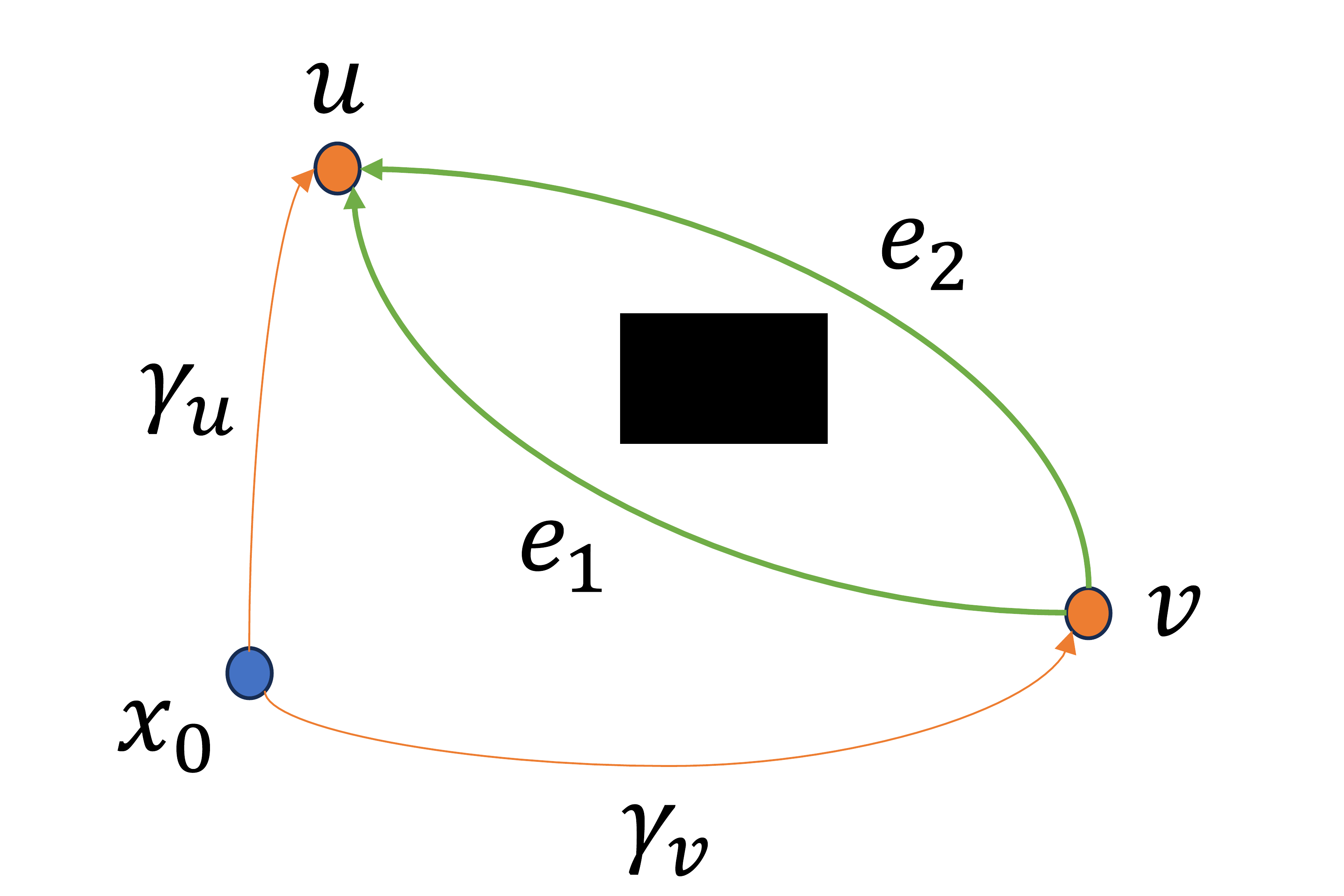}
    \caption{Example of construction of homotopy-augmented graph. The vertices are labeled $u$ and $v$ and edges are labeled $e_1$ and $e_2$. The base point $x_0$ and the paths $\gamma_u$ and $\gamma_v$ are selected to calculate the homotopy classes. The homotopical difference between edges $e_1$ and $e_2$ is represented by different elements $[\gamma_v \circ e_1 \circ \gamma_u^{-1}]$ and $[\gamma_v \circ e_2 \circ \gamma_u^{-1}]$ in $\pi_1(X)$.}
    \label{fig:h-graph-example}
    \Description{One obstacle which gamma v, e2, and gamma u go around and gamma v, e1, and gamma u do not}
\end{figure}

Figure~\ref{fig:h-graph-example} is a simple example. While $\gamma_v\circ e_1 \circ\gamma_u^{-1}$ is homotopic with the constant function, $\gamma_v\circ e_2 \circ\gamma_u^{-1}$ is not due to the existence of the obstacle. Thus, $h(e_1)$ is the identity element, and $h(e_2)$ is the element corresponding to one counterclockwise turn around the obstacle. In the homotopy-augmented graph, two edges above $e_1$ and $e_2$ with the same source have different targets.

\begin{figure}[t]
    \centering
    \includegraphics[width=0.5\linewidth]{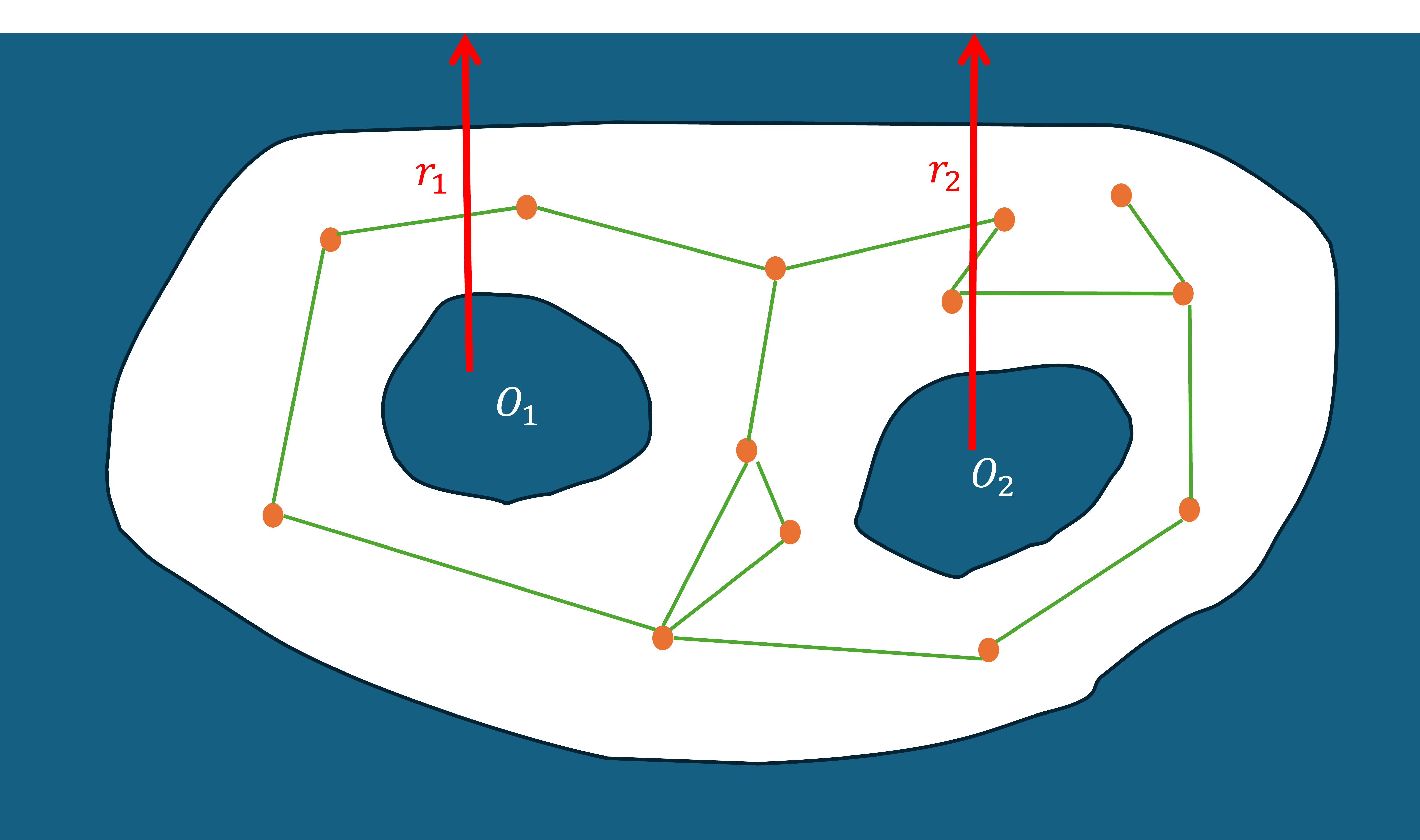}
    \caption{Field with obstacles and roadmap on it. Rays from obstacles are used to calculate the homotopy classes of paths on the free space.}
    \label{fig:rays}
    \Description{Field with two obstacles, roadmaps, and rays}
\end{figure}

Figure~\ref{fig:rays} is another example.
We fix rays which do not pass through any vertex of the roadmap from each obstacle.
We fix an arbitrary point that does not lie on any ray as the base point $x_0$ and take a path that does not intersect any ray as $\gamma_v$ for a vertex $v$.
Then, for an edge $e$, the element $h(e)$ is trivial when $e$ does not intersect any ray.
Let $\rho_i$ be the element of the fundamental group corresponding to the loop that goes around the obstacles $O_i$ without intersecting another ray, for $i=1,2$. Then, $h(e)$ is $\rho_i$ or $\rho_i^{-1}$ when $e$ intersects a ray $r_i$ only once. The sign depends on the direction of intersection.
This approach is equivalent to the classical method for homotopy-aware single-agent path planning on the plane with obstacles~\citep{grigoriev1998polytime,tovar2010sensor}.
In this method, during path planning, we detect crossing rays to calculate the homotopy class.
We construct a string called a \textit{word}, which is explained more generally in the next subsection, by adding a generator $\rho_i$ or its inverse $\rho_i^{-1}$ whenever traversing a ray $r_i$.
This word represents the homotopy class.
A consecutive pair of a generator and its inverse in a word can be canceled.
A word which does not contain such pairs is called \textit{reduced}.
A vertex of the homotopy-augmented graph consists of a vertex of the original graph and a reduced word, which represents a homotopy class uniquely.

\subsection{Presentation of Group}\label{subsec:presentation}
To describe the braid group in the next subsection, we explain the notion of presentation of a group, which is a way to specify a group.

A subset $S$ of a group $G$ is called a \textit{set of generators} of $G$ if any element of $G$ can be written as a finite product of elements of $S$ and their inverses.
An element of $S$ is called a \textit{generator}.
A \textit{word} is a string consisting of generators and their inverses, which represents an element of the group.
A \textit{relation} is a word which represents the identity element.

Intuitively, a \textit{presentation} $\langle S|R\rangle$ of a group $G$ consists of a set $S$ of generators and a set $R$ of relations which is sufficient to specify $G$. Formally, $\langle S|R\rangle$ is a presentation of $G$ if the following conditions are satisfied.
\begin{itemize}
\item $S$ is a set of generators.
\item The kernel of the natural surjection from the free group $F_S$ on $S$ to $G$ is the smallest normal subgroup of $F_S$ containing $R$.
\end{itemize}

A relation $ab^{-1}\in R$ is sometimes denoted as $a=b$.

The decision problem of determining whether two given words represent the same element of the group is known as the \textit{word problem}~\citep{peifer1997introduction}.

For example, the fundamental group for the example in Figure~\ref{fig:rays} has a presentation $\langle \rho_1, \rho_2| \emptyset \rangle$. Since there exists no relation between $\rho_1$ and $\rho_2$, the set of relations is empty.
For another example, if a group $G$ has a presentation $\langle a,b| aba^{-1}b^{-1}\rangle$, then $G$ is isomorphic to $\mathbb{Z}\times\mathbb{Z}$ because any element is written as $a^nb^m$ uniquely with $n,m \in \mathbb{Z}$ uniquely.

\subsection{Braid Group}\label{subsec:braid}
For a domain $\mathcal{D}$, the \textit{configuration space} $\mathcal{C}_n(\mathcal{D})$ is defined as
\begin{equation}
\mathcal{C}_n(\mathcal{D}) := \left\{(p_1,p_2,\ldots,p_n)\in \mathcal{D}^n\ \middle|\ p_i\neq p_j\text{ for all }i\neq j\right\}.
\end{equation}
As the name suggests, this space is the configuration space for the multi-agent path planning of $n$ agents without size in $\mathcal{D}$, because the excluded area corresponds to states in which some two agents occupy the same position.

The $n$-th symmetric group $S_n$ naturally acts on $\mathcal{C}_n(\mathcal{D})$ by permuting the indices. The \textit{unlabeled configuration space} $\mathcal{UC}_n(\mathcal{D})$ is defined as the quotient of $\mathcal{C}_n(\mathcal{D})$ by $S_n$. This space is the configuration space for unlabeled multi-agent path planning. In this problem setting, agent indexes are ignored, starts and goals do not correspond, and each agent may reach any goal, while exactly one agent must reach at each goal.

The fundamental group of $\mathcal{UC}_n(\mathbb{R}^2)$ is known as the $n$-th \textit{braid group} $B_n$. The fundamental group of $\mathcal{C}_n(\mathbb{R}^2)$, which is a subgroup of $B_n$, is known as the $n$-th \textit{pure braid group} $P_n$. There exists a natural surjective group homomorphism from $B_n$ to $S_n$, the kernel of which is $P_n$.

$B_n$ has a presentation given by
\begin{equation}
\langle \sigma_1,\ldots,\sigma_{n-1}\ |\ \sigma_i\sigma_j\sigma_i^{-1}\sigma_j^{-1},\ \sigma_i\sigma_{i+1}\sigma_{i}\sigma_{i+1}^{-1}\sigma_{i}^{-1}\sigma_{i+1}^{-1}\rangle,
\end{equation}
where $1\leq i < n-1$, and $i+1<j\leq n-1$~\citep{rolfsen2010tutorial}.
See \S~\ref{subsec:word construction} for the geometrical description of this presentation.

\section{Problem Setting}\label{subsec:problem setting}

In this section, we describe the setting for homotopy-aware path planning in planar domains.

\begin{figure}[t]
\centering
    \includegraphics[width=0.5\columnwidth]{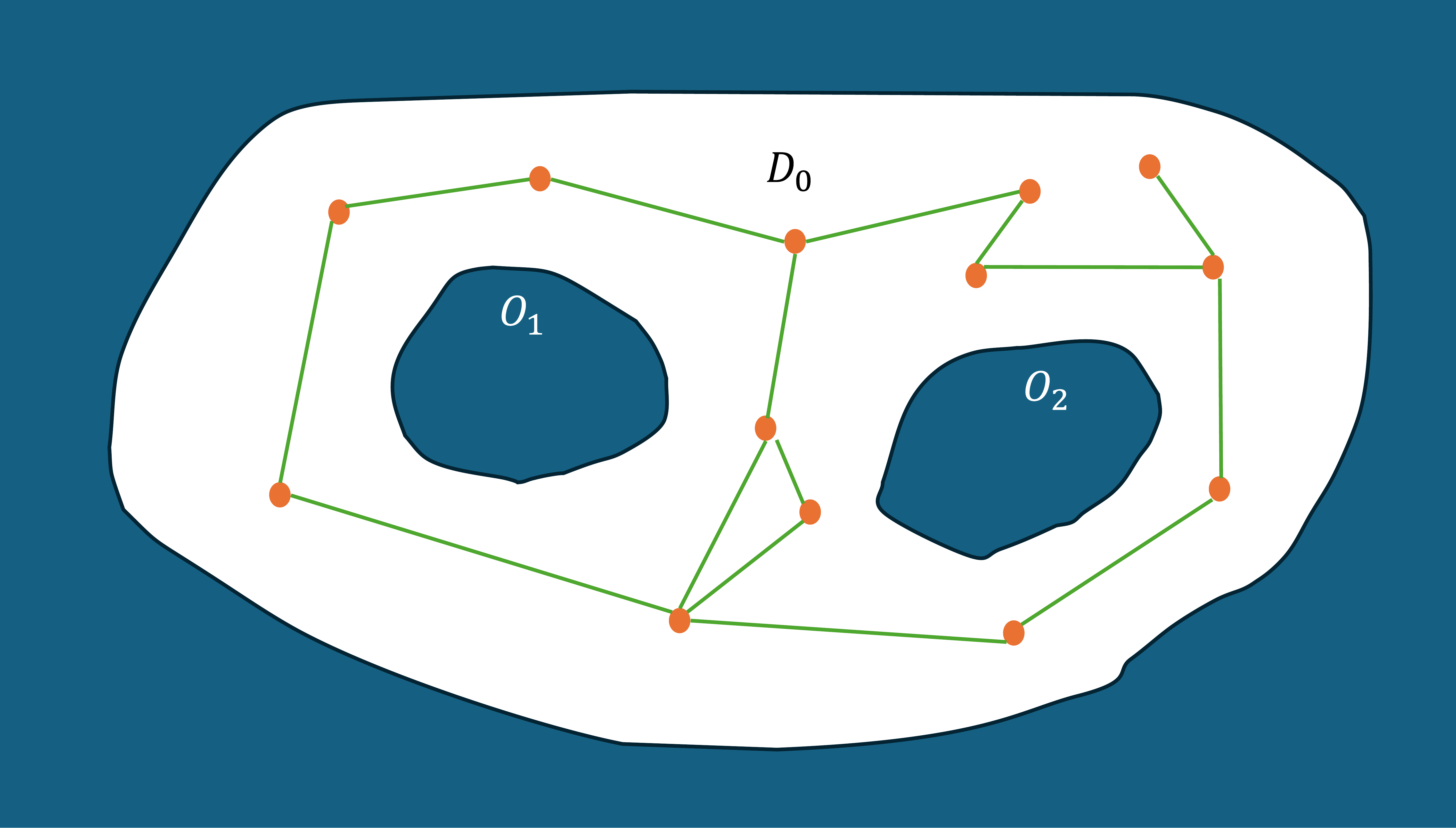}
\caption{Example of planar domain $\mathcal{D}$ satisfying the assumption of the problem setting and roadmap $G$ over $\mathcal{D}$. Here, $\mathcal{D} = \mathcal{D}_0\setminus (O_1 \cup O_2)$, where $\mathcal{D}_0$ is the region enclosed by the outer boundary and $O_1$ and $O_2$ are obstacles.}
\label{fig:field}
\Description{Field with two obstacles}
\end{figure}

We consider homotopy-aware multi-agent path planning with $n$ agents in the domain $\mathcal{D}\subseteq \mathbb{R}^2$.
As illustrated in Figure~\ref{fig:field},
we assume that $\mathcal{D}$ has the form $D_0\setminus(O_1\cup \cdots \cup O_r)$, where $D_0$ and $O_1,\ldots,O_r$ (obstacles) are bounded domains whose boundaries are simple closed curves and $\mathbb{R}^2\setminus D_0, O_1,\ldots,O_r$ neither intersect nor touch.\footnote{Strictly speaking, for example, a domain in a grid with obstacles touching at corners does not satisfy this assumption. However, such obstacles can be considered as one obstacle by widening them slightly at the corners.}

We assume that a finite undirected planar graph $G=(V,E)$ lying on the interior of $\mathcal{D}$ is given as a roadmap.
Starts $s_1,\ldots,s_n\in V$ and goals $g_1,\ldots,g_n\in V$ are also given.
Time is discretized into steps, and each agent either moves along an edge or remains at a vertex in each step.
Namely, a single-agent plan for the $i$-th agent is a sequence $v_{i,0}, v_{i,1},\ldots, v_{i,T} \in V$ of vertices such that $v_{i,0}=s_i$, $v_{i,T}=g_i$, and $v_t=v_{t+1}$ or $(v_t, v_{t+1}) \in E$ for any $t$.
A \textit{solution} is a tuple of single-agent plans for all agents without any vertex or swapping conflicts, i.e., $v_{i,t}\neq v_{j,t}$ and $(v_{i,t},v_{i,t+1}) \neq (v_{j,t+1}, v_{j,t})$ for all $i\neq j$~\citep{stern2019multi}. The \textit{sum of costs} of a solution is defined as the sum of the reaching times for all agents, which means
\begin{equation}
\sum_{i=1}^n \min\{t\mid v_{i,t}=g_i\}.
\end{equation}

A solution determines a path in the configuration space $\mathcal{C}_n(\mathcal{D})$.
Two solution are called \textit{homotopic} if their corresponding paths are homotopic.
When $K>0$ is given, 
We want to find $K$ mutually non-homotopic solutions with small sums of costs.

\begin{figure}[t]
\centering
    \begin{subfigure}{0.19\columnwidth}
        \includegraphics[width=0.96\columnwidth]{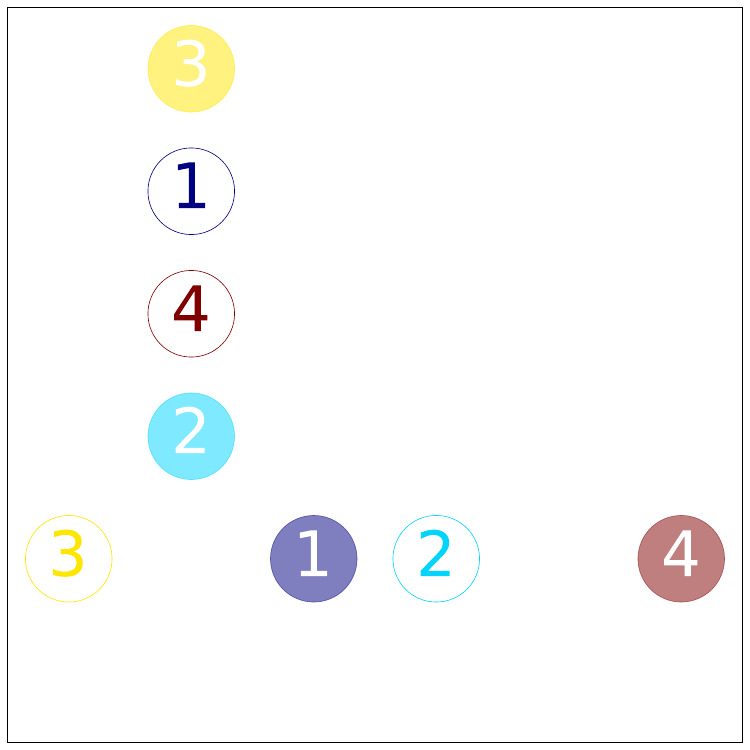}
        \caption{Instance}
    \end{subfigure}    
    \hspace{10pt}
    \begin{subfigure}{0.76\columnwidth}
        \includegraphics[width=0.24\columnwidth]{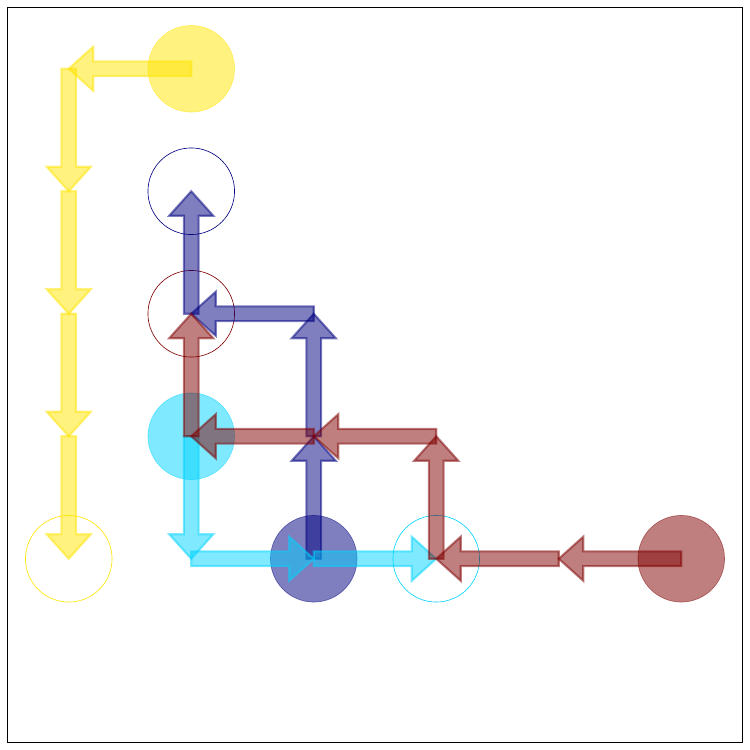}
        \includegraphics[width=0.24\columnwidth]{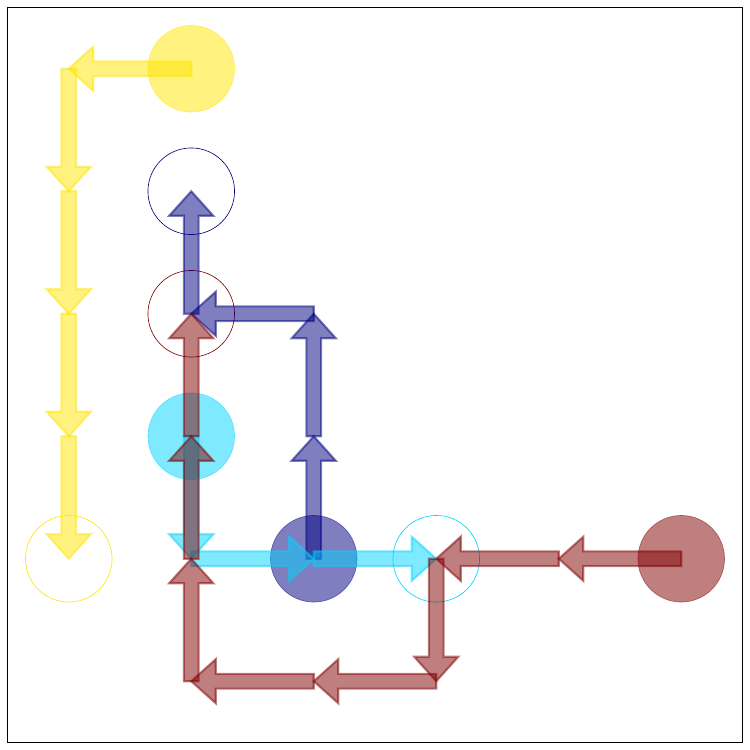}
        \includegraphics[width=0.24\columnwidth]{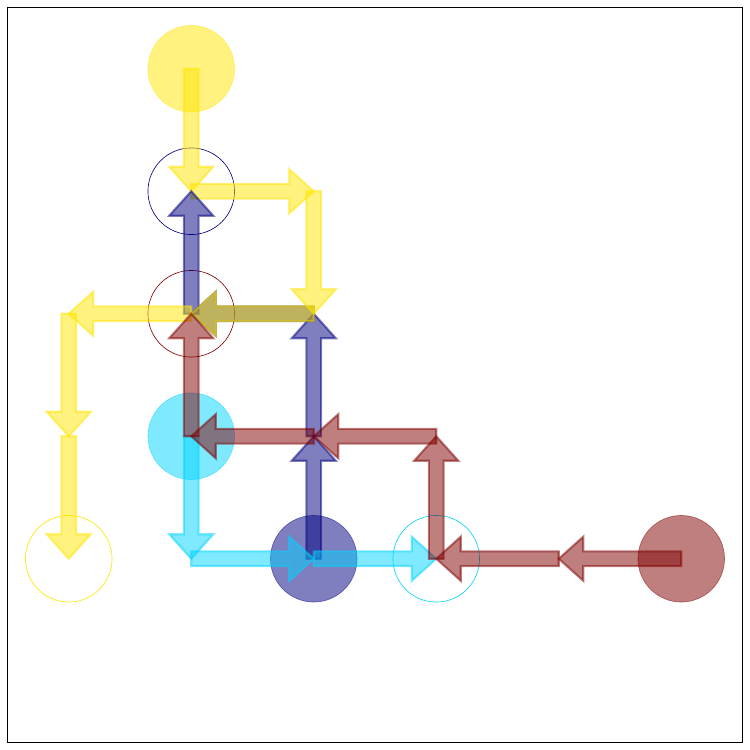}
        \includegraphics[width=0.24\columnwidth]{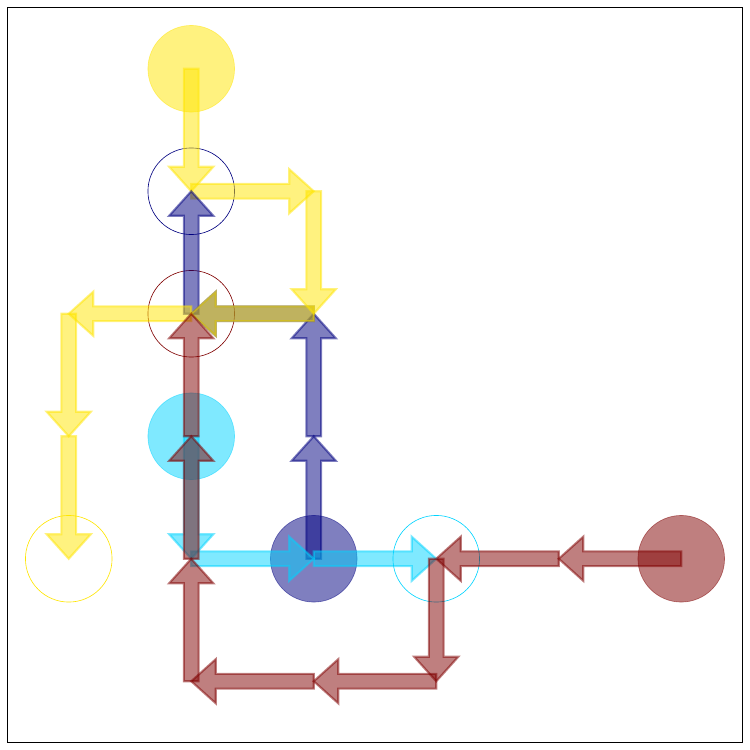}
        \caption{Solutions}
    \end{subfigure}
\caption{Examples of Instance of MAPF on grid and its homotopically distinct solutions. The filled circles are for the starts, and the unfilled circles with the same indexes are for the corresponding goals.
The moves of agents are represented by arrows.
In these solutions, no agents stay at the same vertex except at their goals. The displayed four solutions belong to different homotopy classes.
More specifically, we can classify them based on which side the third agent avoids the first agent and which side the fourth agent avoids the second agent.}
\label{fig:example_instance}
\Description{An instance figure with four starts and goals on a grid map and four solutions figures with arrows on it}
\end{figure}

Figure~\ref{fig:example_instance} shows an instance and its homotopically distinct solutions in a grid world.

As stated in \S~\ref{sec:preliminary}, the problem can be reduced to pathfinding on the homotopy-augmented graph in $\mathcal{C}_n(\mathcal{D})$.

\begin{remark}
The problem of finding $K$ mutually non-homotopic paths with minimum lengths is called \textit{$K$-shortest non-homotopic path planning}~\citep{bhattacharya2018path,yang2022efficient}.
Since our priority in this paper is scalability, outputs of our method are not necessarily optimal. See also Remark~\ref{remark:k-SNHP}.
\end{remark}

\begin{remark}\label{rem:continous_time}
For simplicity, we consider the case where $G$ has edges of the equal cost with no crossing, and where only vertex and swapping conflicts are forbidden.
Our framework works even when time is continuous or collision conditions are complex~\citep{kasaura2022prioritized}, because these conditions only affect the construction of the timed graph in Line~\ref{line:construct_graph} in Algorithm~\ref{algorithm:HRPP}.
However, such collision conditions are ignored in the homotopy calculations. See also the next remark.
\end{remark}

\begin{remark}
Since the homotopy classes are defined in the configuration space, intermediate paths do not need to correspond to any solutions in $G$.
The agents are considered as points during continuous deformation of solutions.
Note that this simplification only affects homotopies; sizes of agents are considered in pathfinding.
Ignoring the sizes of agents does not matter in cases with no obstacles because the agents can be far enough apart from each other when deforming paths to other ones. On the other hand, when there exist obstacles, two paths that are actually homotopically distinct when both sizes of agents and obstacles are taken into account may be considered as homotopic. See also \S~\ref{sec:conclusion}.
However, the fact that our homotopical classification is not complete does not mean that it is useless in such cases.
\end{remark}

\section{Method}

In this section, we describe our method to generate multiple solutions with different homotopies.

The entire algorithm, which is provided in \S~\ref{subsec:revised prioritized planning}, is a combination of revised prioritized planning and our framework for calculating homotopy classes while planning.
In the first four subsections, we describe the framework.
For the setting where agents are not labeled and there are no obstacles, a classical method to construct words representing homotopy classes is known, which is described in \S~\ref{subsec:word construction}.
In \S~\ref{subsec:reduction to obstacle-free case} and \S~\ref{subsec:reduction to unlabeled case}, we reduce our problem setting to this case.
Since constructed words themselves are not suitable for comparison, a method to represent them by tuple of integers, which is called Dynnikov coordinates, is explained in \S~\ref{subsec:Dynnikov Coordinate}.
In \S~\ref{subsec:completeness}, we prove the completeness of our algorithm under specific assumptions.

\subsection{Reduction to Obstacle-Free Case}\label{subsec:reduction to obstacle-free case}

\begin{figure}[t]
\centering
    \includegraphics[width=0.5\columnwidth]{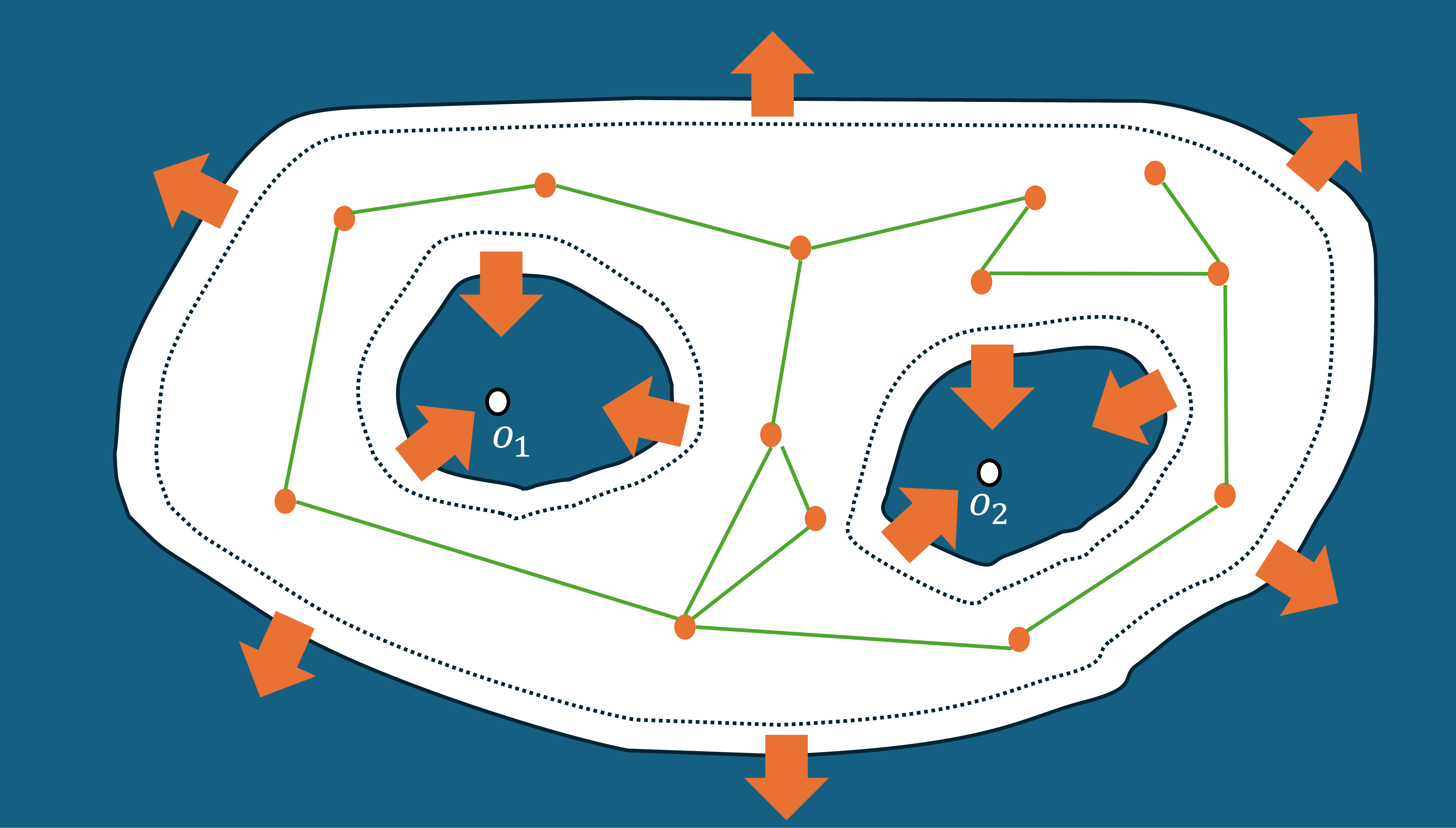}
\caption{Deformation of $\mathcal{D}=\mathcal{D}_0\setminus(O_1\cup O_2)$ to $\mathbb{R}^2\setminus \{o_1, o_2\}$, where $o_1\in O_1$ and $o_2\in O_2$.
Regions near obstacles expand to collapse obstacles to points.
A region near the outer boundary expands to the point at infinity.
Points in the region enclosed by dotted curves, which includes roadmap $G$, are fixed during deformation.}
\label{fig:deform}
\Description{Field in the previous figure with arrows to oi and to the outside}
\end{figure}

We take one representative element $o_i\in O_i$ for each obstacle.
Obviously, $o_1,\ldots,o_r$ are different from each other.
By taking sufficiently small neighborhoods of boundaries and transforming them,
as illustrated in Figure~\ref{fig:deform},
we can deform $\mathcal{D}$ to $\mathbb{R}^2\setminus\{o_1,\ldots,o_r\}$ without changing the roadmap.
For a rigorous proof, the Jordan–Sch\"{o}nflies theorem~\citep{thomassen1992jordan} can be used.
Thus, it is enough to calculate homotopies on $\mathbb{R}^2\setminus\{o_1,\ldots,o_r\}$. Note that, since the roadmap is not changed, the found paths are lying on the original $\mathcal{D}$.

We reduce the calculation of homotopy in $\mathcal{C}_n(\mathbb{R}^2\setminus\{o_1,\ldots,o_r\})$ to that in $\mathcal{C}_{r+n}(\mathbb{R}^2)$. Intuitively, we consider the obstacles
as additional agents staying at the same positions. While it is non-trivial that this reduction preserves homotopy classes because obstacles cannot move when deforming paths, the following proposition guarantees it.

\begin{proposition}\label{prop:homotopy_inj}
    The map
\begin{equation}        
    \pi_1\left(\mathcal{C}_n(\mathbb{R}^2\setminus\{o_1,\ldots,o_r\})\right)\to \pi_1\left(\mathcal{C}_{r+n}(\mathbb{R}^2)\right)
\end{equation}
    induced by the embedding
\begin{equation}\label{eq:embedding}
\begin{array}{ccc}
\mathcal{C}_n(\mathbb{R}^2\setminus\{o_1,\ldots,o_r\}) & \hookrightarrow & \mathcal{C}_{r+n}(\mathbb{R}^2) \\
\rotatebox[origin=c]{90}{$\in$} & & \rotatebox[origin=c]{90}{$\in$}\\
(p_1,\ldots, p_n) & \mapsto &(o_1,\ldots, o_r, p_1, \ldots, p_n)
\end{array}
\end{equation}
    is injective.
\end{proposition}
See Appendix~\ref{appendix:proof} for the proof.

In the following three subsections, we replace $r+n$ with $n$ and consider $\mathcal{C}_{n}(\mathbb{R}^2)$ for simplicity.
We also write simply $\mathcal{C}_{n}$ and $\mathcal{UC}_{n}$ instead of $\mathcal{C}_{n}(\mathbb{R}^2)$ and $\mathcal{UC}_{n}(\mathbb{R}^2)$.

\subsection{Reduction to Unlabeled Case}\label{subsec:reduction to unlabeled case}
As mentioned earlier, in labeled multi-agent path planning, we can calculate homotopy classes as unlabeled multi-agent path planning. To formalize this reduction, let $G^n=(V^n,E^n)$ be the $n$-times direct product of $G$ minus collision parts, which is a graph on the configuration space $\mathcal{C}_n$. 
Instead of directly applying the construction in \S~\ref{subsec:h-graph} to $\mathcal{C}_n$,
we construct a graph $(V^n_h, E^n_h)$ as follows.

\begin{itemize}
\item $V^n_h := V^n\times \pi_1(\mathcal{UC}_n) = V^n\times B_n$.
\item For each edge $e\in E^n$ from vertex $v\in V^n$ to vertex $u \in V^n$, and for each element $\alpha \in B_n$, $E^n_h$ contains an edge from $(v,\alpha)$ to $(u,\alpha h(\overline{e}))$, where $\overline{e}$ is the projection of $e$ to $\mathcal{UC}_n$ by natural surjection, and $h(\overline{e})$ is calculated using the method described in \S~\ref{subsec:word construction}.
\end{itemize}

It is important to note that, for a specific multi-agent path planning task, it is not necessary to construct the entire graph as it is not connected. Instead, we can focus on the relevant connected components of the graph. For example, for any vertex $v\in V^n$ and elements $\alpha, \beta\in B_n$, the connected components of $(v,\alpha)$ and $(v,\beta)$ are different unless $\alpha^{-1}\beta\in P_n$.

With our approach, we represent the homotopy classes of solutions for a specific instance of the labeled multi-agent path planning problem by labeling them with elements in a coset of the pure braid group, instead of those in the pure braid group. The choice of coset depends on the relative positions of the agents' start and goal locations.
For a detailed explanation of why we do not use the pure braid group, see Appendix~\ref{appendix:pn}.

\subsection{Word Construction}\label{subsec:word construction}
The construction of words for $\mathcal{UC}_n$ is classically known~\citep{fox1962braid}.

We fix a coordinate $x,y$ for $\mathbb{R}^2$.
Let $p_1=(x_1,y_1),\ldots,p_n=(x_n,y_n)$ be the induced coordinates of $\mathcal{C}_n$. We define:
\begin{equation}
    \lambda_{2n} := \{x_1< x_2< \cdots < x_n\},
\end{equation}
\begin{equation}
    \lambda^{2n-1}_{i} := \{x_1< \cdots < x_i=x_{i+1}< \cdots < x_n, y_i<y_{i+1}\},
\end{equation}
for $1\leq i < n$. It is easy to see that the map $\lambda^{2n}\cup\bigcup_i \lambda^{2n-1}_i \hookrightarrow \mathcal{C}_n \twoheadrightarrow \mathcal{UC}_n$ is injective. We identify $\lambda^{2n}, \lambda^{2n-1}_i$ with their images.
$\mathcal{UC}_n\setminus\lambda^{2n}$ is $(2n-1)$-dimensional and $\mathcal{UC}_n\setminus\lambda^{2n}\setminus \bigcup_i \lambda^{2n-1}_i$ is $(2n-2)$-dimensional.

For simplicity, we assume that all $v\in V^n$ lie on $\lambda^{2n}$. We also choose a base point $x_0$ of the fundamental group within $\lambda^{2n}$. To make the homotopy-augmented graph explicit, we choose the standard path connecting $x_0$ and vertices of $V^n$ in $\lambda^{2n}$.
The generator $\sigma_i$ of $\pi_1(\mathcal{UC}_n)=B_n$ is represented by a loop in $\lambda^{2n}\cup\lambda^{2n-1}_i$ that traverses $\lambda^{2n-1}_i$ exactly once with the direction from region $\{y_i<y_{i+1}\}$ to region $\{y_i>y_{i+1}\}$. Consequently, for an edge $e$ on $\mathcal{UC}_n$, the word representing the corresponding braid $h(e)$ is constructed by detecting the traverses of $e$ with $\lambda^{2n-1}_1,\ldots,\lambda^{2n-1}_{n-1}$.
By examining the intersections of the boundaries of $\lambda^{2n-1}_1,\ldots,\lambda^{2n-1}_{n-1}$, we obtain the relations $\sigma_i\sigma_j=\sigma_j\sigma_i$ for $i+1<j$ and $\sigma_i\sigma_{i+1}\sigma_i=\sigma_{i+1}\sigma_i\sigma_{i+1}$. Detailed derivations of these relations are omitted for brevity.

Intuitively, this construction can be described as follows. We number the agents in ascending order of their $x$-coordinates. $\sigma_i$ corresponds to the counterclockwise swap (Figure~\ref{fig:counterclockwise}) of the agent with number $i$ and the agent with number $(i+1)$, and $\sigma_i^{-1}$ corresponds to their clockwise swap (Figure~\ref{fig:clockwise}).
The word for an edge is constructed by detecting the swaps of agent indexes.
The graphical explanation for the relations is shown in Figure~\ref{fig:word relations}.
The key observation is that the order in which the agents move does not affect the homotopy class. Therefore, we have two relations:
(a) $\sigma_i\sigma_j=\sigma_j\sigma_i$ and
(b) $\sigma_i\sigma_{i+1}\sigma_i=\sigma_{i+1}\sigma_i\sigma_{i+1}$.

Figure~\ref{fig:braid example} shows two examples of word construction.

\begin{figure}[t]
\centering
    \centering
    \begin{subfigure}{0.3\columnwidth}
        \includegraphics[width=\textwidth]{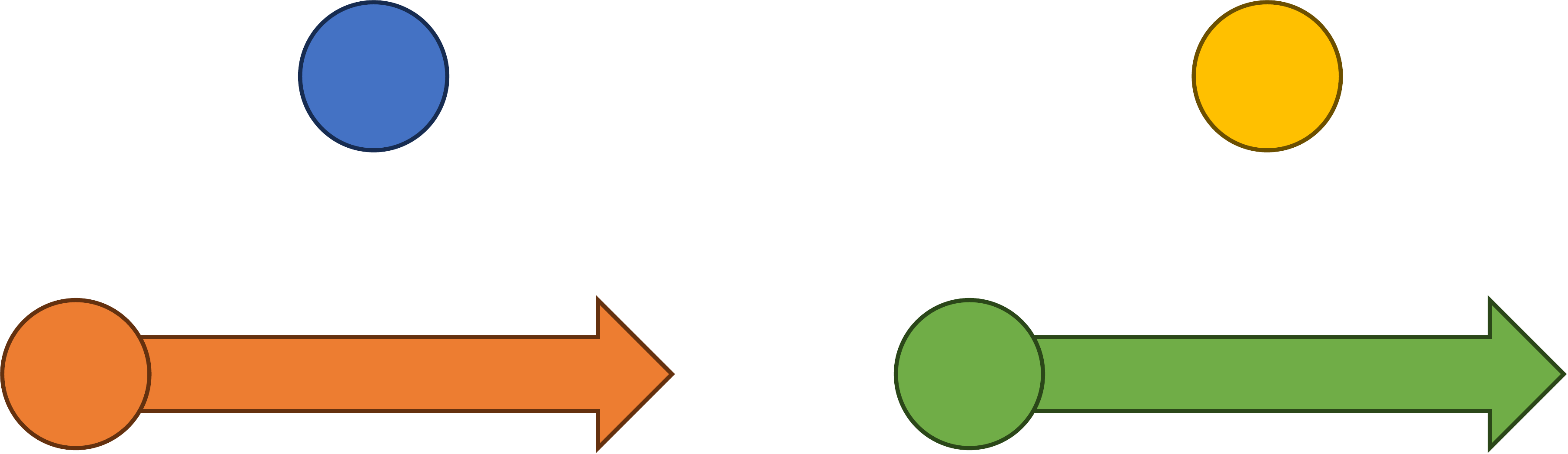}
        \caption{$\sigma_i\sigma_j=\sigma_j\sigma_i$}
    \end{subfigure}
    \hspace{30pt}
    \centering
    \begin{subfigure}{0.24\columnwidth}
        \includegraphics[width=\textwidth]{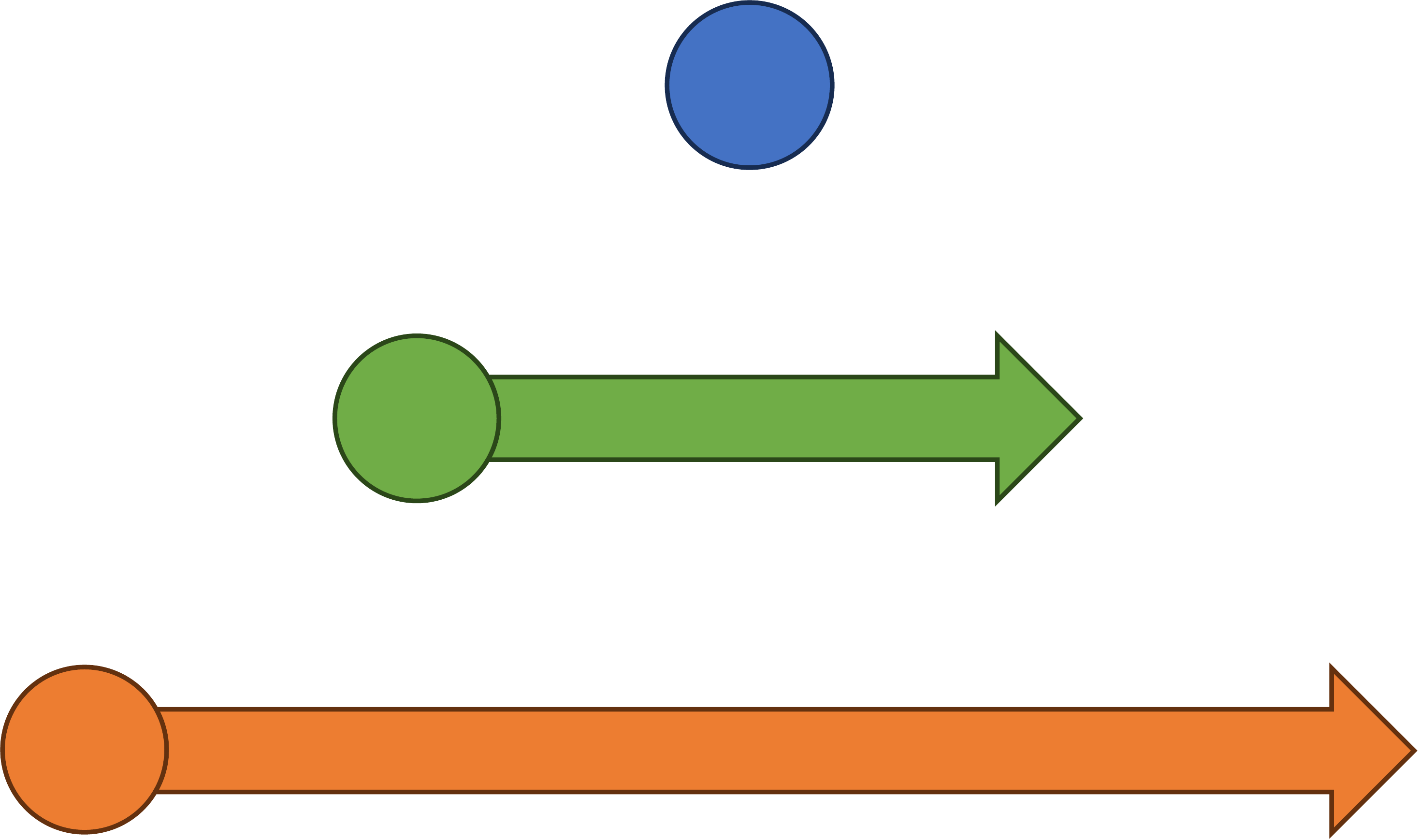}\caption{$\sigma_i\sigma_{i+1}\sigma_i=\sigma_{i+1}\sigma_i\sigma_{i+1}$}
    \end{subfigure}
\caption{Figures illustrating relations between the generators.
In (a), the left circles represent the agents with numbers $i$ and $(i+1)$, and the right circles represent the agents with numbers $j$ and $(j+1)$, where $j>i+1$. In (b), the bottom, middle, and top circles represent the agents with numbers $i$, $(i+1)$, and $(i+2)$, respectively.
Arrows represent movements.
The homotopy being independent of the order of moves corresponds to the relations.}
\label{fig:word relations}
\Description{Horizontally aligned two vertically aligned a point and a leftarrow. Vertically aligned a point, short leftarrow, and long leftarrow}
\end{figure}
\begin{figure}[t]
\centering
    \includegraphics[width=0.7\textwidth]{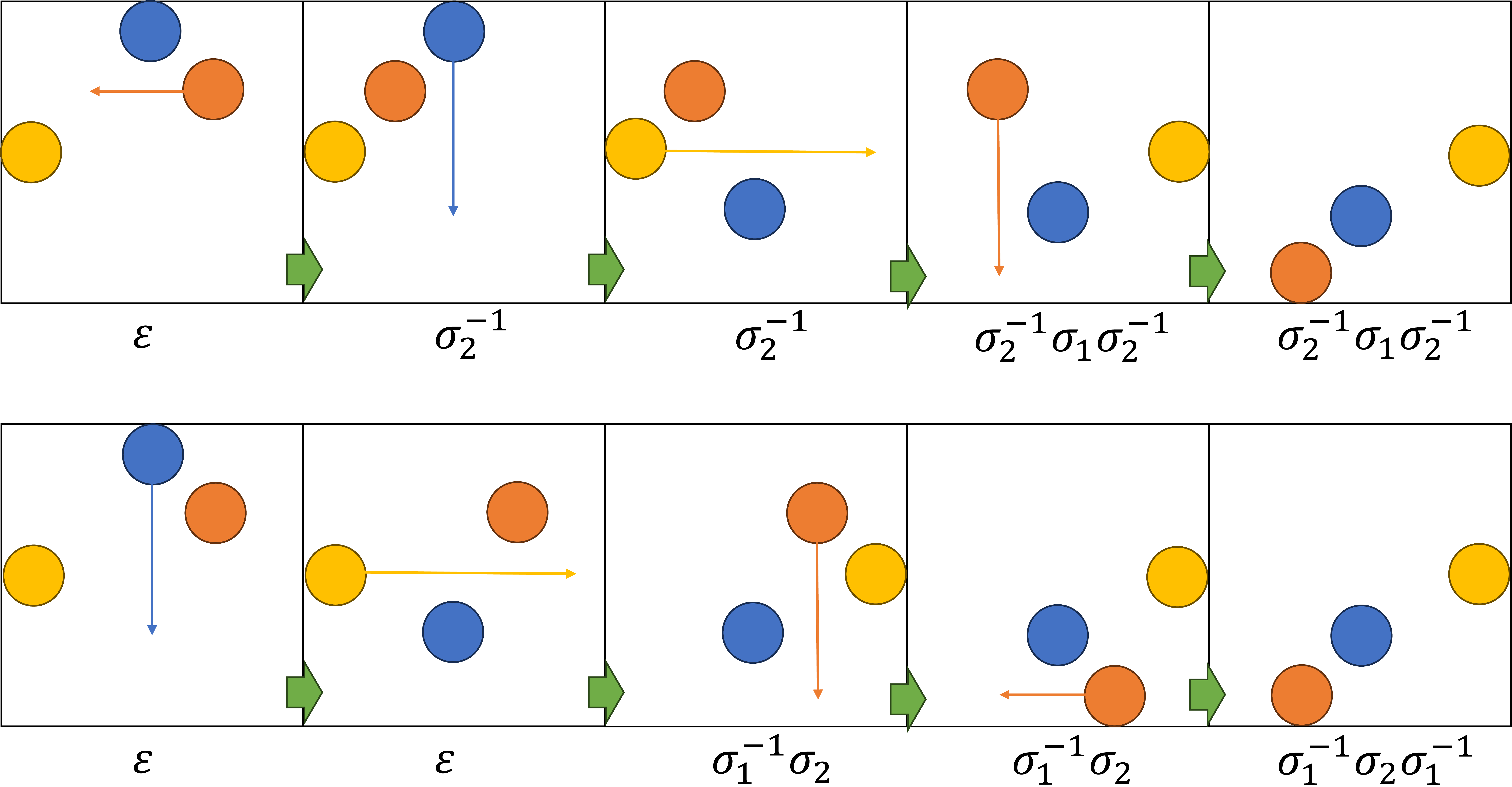}
\caption{Two examples of braid construction. From left to right, generators are added or not added to the displayed braid word as an agent moves.
In the top example, the agents with numbers $2$ and $3$ swap clockwise at the first block, then the agents with numbers $1$ and $2$ swap counterclockwise and the agents with numbers $2$ and $3$ swap clockwise at the third block. In the example below, the agents with numbers $1$ and $2$ swap clockwise and the agents with numbers $2$ and $3$ swap counterclockwise at the second block; then, the agents with numbers $1$ and $2$ swap clockwise at the fourth block.}
\label{fig:braid example}
\Description{Solutions for MAPF on grid and words}
\end{figure}

\begin{remark}
The word construction depends on the choice of the coordinates $x$ and $y$, because the choice of generators depends on them.
However, since homotopy itself is independent of the coordinates, if braids for two paths with the same endpoints calculated by using some coordinates are the same, their braids always coincide even when other coordinates are used.
\end{remark}

\subsection{Dynnikov Coordinates}\label{subsec:Dynnikov Coordinate}

\begin{figure}[t]
    \centering
    \includegraphics[width=0.7\linewidth]{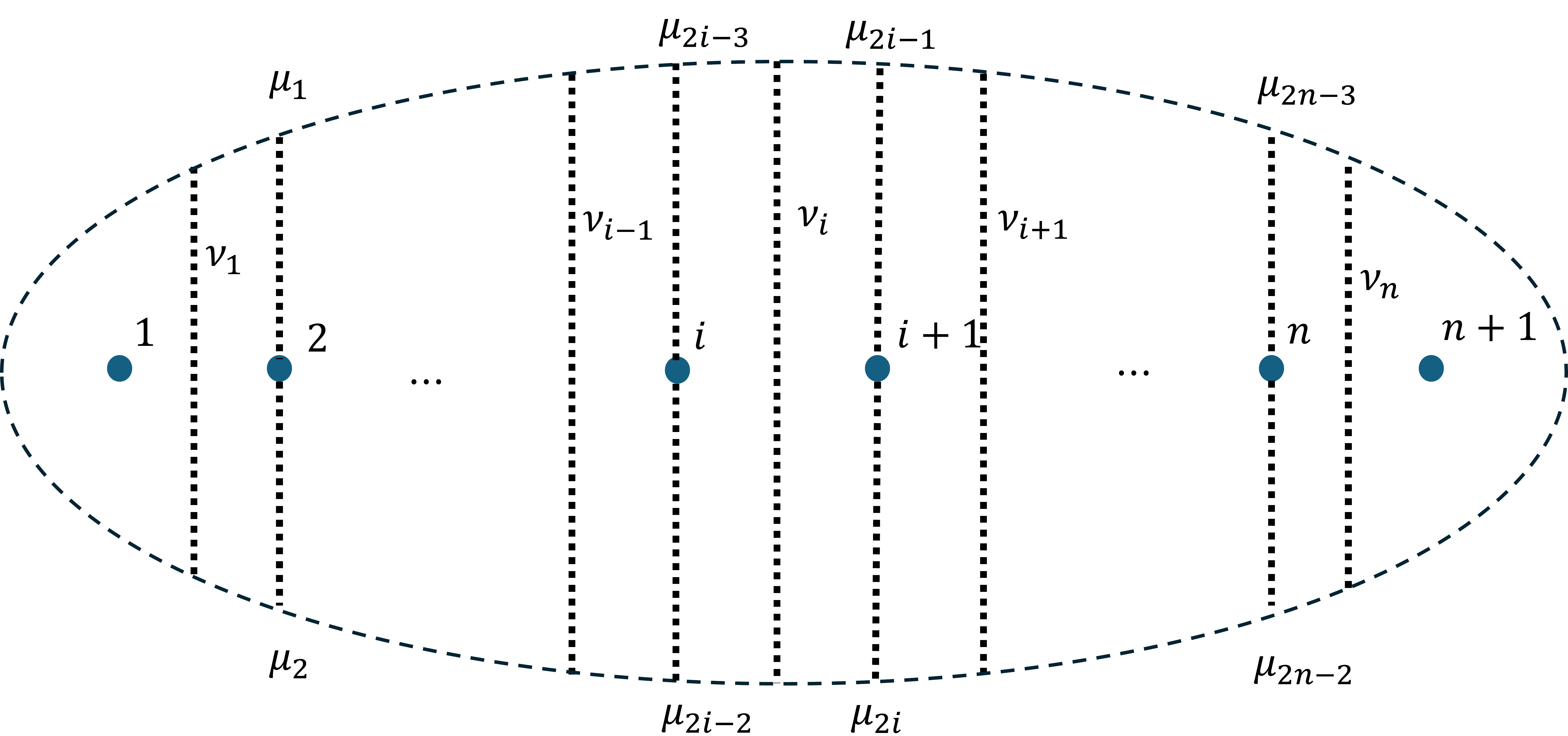}
    \caption{Plane with $(n+1)$ holes to calculate Dynnikov coordinates.
    For representation of multicurves on it by integers, the dotted lines or half lines are used to count the numbers of times multicurves intersect them.}
    \label{fig:Dynnikov_plane}
    \Description{Horizontally aligned points and vertical lines in an ellipse}
\end{figure}
\begin{figure}[t]
    \centering
    \includegraphics[width=0.5\linewidth]{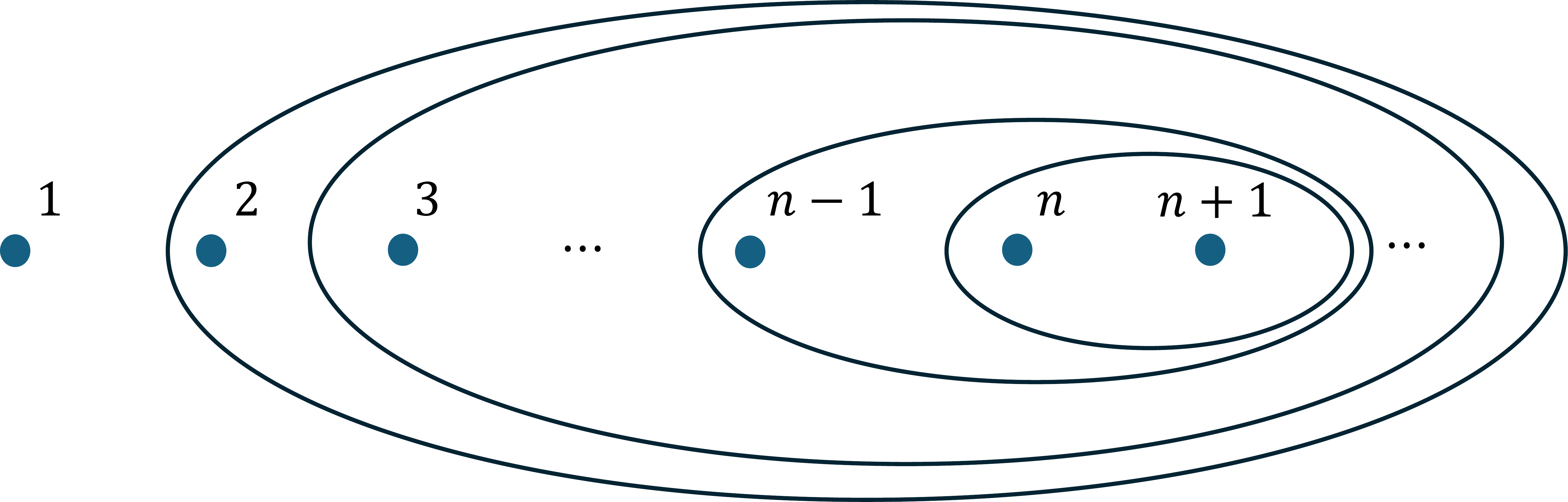}
    \caption{Initial multicurve to solve the word problem for braids with Dynnikov coordinates.}
    \label{fig:Dynnikov_u}
    \Description{Nested curves}
\end{figure}

We use Dynnikov coordinates to uniquely represent braids.
In brief, we represent the braid group as actions to tuples of integers, which can be calculated only by using addition, subtraction, maximum, and minimum.
We follow the description in \citet{thiffeault2022braids}.

Intuitively, Dynnikov coordinates represent a multicurve on a plane with aligned $(n+1)$ holes.\footnote{While $n$ holes are enough to define the action by $B_n$, this action is not faithful.} Here, ``multicurve" means a disjoint union of simple closed curves, any of which encircles at least two holes but not all of them.
For an example, see the example calculation at the end of this subsection.
Figure~\ref{fig:Dynnikov_plane} illustrates a plane with holes, where symbols attached to lines or half lines indicate the number of times the multicurve intersects them.
Holes are numbered from left to right.
For $2\leq i \leq n$, $\mu_{2i-3}$ and $\mu_{2i-2}$ correspond to the half lines extending up and down from the $i$-th hole, respectively.
For $1\leq i \leq n$, $\nu_i$ corresponds to the vertical line between the $i$-th and $(i+1)$-th holes.
When counting the numbers of intersections, we take the minimum ones, allowing for homotopical deformation of curves.

For any $1\leq i \leq n-1$, let
\begin{align*}
    a_i&:=\frac{\mu_{2i}-\mu_{2i-1}}{2},\\
    b_i&:=\frac{\nu_i-\nu_{i+1}}{2}.
\end{align*}
Then, $(a_1,\ldots,a_{n-1},b_1,\ldots,b_{n-1})$ gives a bijection between homotopy classes of multicurves and $\mathbb{Z}^{2n-2}\setminus\{\mathbf{0}\}$, that is \textit{Dynnikov coordinates}~\citep{hall2009topological}.

The braid group $B_{n+1}$ acts on the homotopy classes of multicurves by moving holes.
For an example, again, see the example calculation at the end of this subsection.
Here, we are only interested in the action of $B_n$ as a subgroup of $B_{n+1}$.
The action of $B_n$ can be written using Dynnikov coordinates as follows.
For $x\in\mathbb{Z}$, we write $x^{+}:=\max\{x,0\}$ and $x^{-}:=\min\{x,0\}$.
For $i=1,\ldots,n-1$ and $e=\pm 1$,
we can calculate $(a_1',\ldots,a_{n-1}',b_1',\ldots, b_{n-1}'):=(a_1,\ldots,a_{n-1},b_1,\ldots, b_{n-1}) \cdot \sigma_i^{e}$ as
\begin{equation}\label{eq:Dyn_update}
(a_k',b_k'):=\begin{cases}
    (-b_1+(a_1+b_1^{+})^{+}, a_1+b_1^{+}) & \text{ for } i=k=1,\,e=+1, \\
    (b_1-(b_1^{+}-a_1)^{+}, b_1^{+}-a_1) & \text{ for } i=k=1,\,e=-1, \\
    (a_{i-1}-b_{i-1}^{+}-(b_i^{+}+c)^{+}, b_i+c^{-}) & \text{ for } i>1,\,k=i-1,\,e=+1, \\
    (a_{i-1}+b_{i-1}^{+}+(b_i^{+}-d)^{+}, b_i-d^{+}) & \text{ for } i>1,\,k=i-1,\,e=-1, \\
    (a_i-b_i^{-}-(b_{i-1}^{-}-c)^{-}, b_{i-1}-c^{-}) & \text{ for } i>1,\,k=i,\,e=+1, \\
    (a_i+b_i^{-}+(b_{i-1}^{-}+d)^{-}, b_{i-1}+d^{+}) & \text{ for } i>1,\,k=i,\,e=-1, \\
    (a_k,b_k) & \text{ for } k\neq i-1,i,
\end{cases}
\end{equation}
where $c:=a_{i-1}-a_i-b_i^{+}+b_{i-1}^{-}$ and $d:=a_{i-1}-a_i+b_i^{+}-b_{i-1}^{-}$~\citep{thiffeault2022braids}.

Moreover, where $u := (a_1=0,\ldots,a_{n-1}=0, b_1=-1,\ldots,b_{n-1}=-1)\in \mathbb{Z}^{2n-2}$, which corresponds to the multicurves in Figure~\ref{fig:Dynnikov_u}, for any $\alpha, \beta \in B_n$, $u\cdot \alpha=u\cdot \beta$ if and only if $\alpha=\beta$~\citep{dynnikov2002yang,thiffeault2022braids}. Thus, we can represent an element $\alpha$ in $B_n$ by $u\cdot \alpha\in \mathbb{Z}^{2n-2}\setminus\{\mathbf{0}\}$.

\begin{figure}[t]
    \centering
    \includegraphics[width=0.9\linewidth]{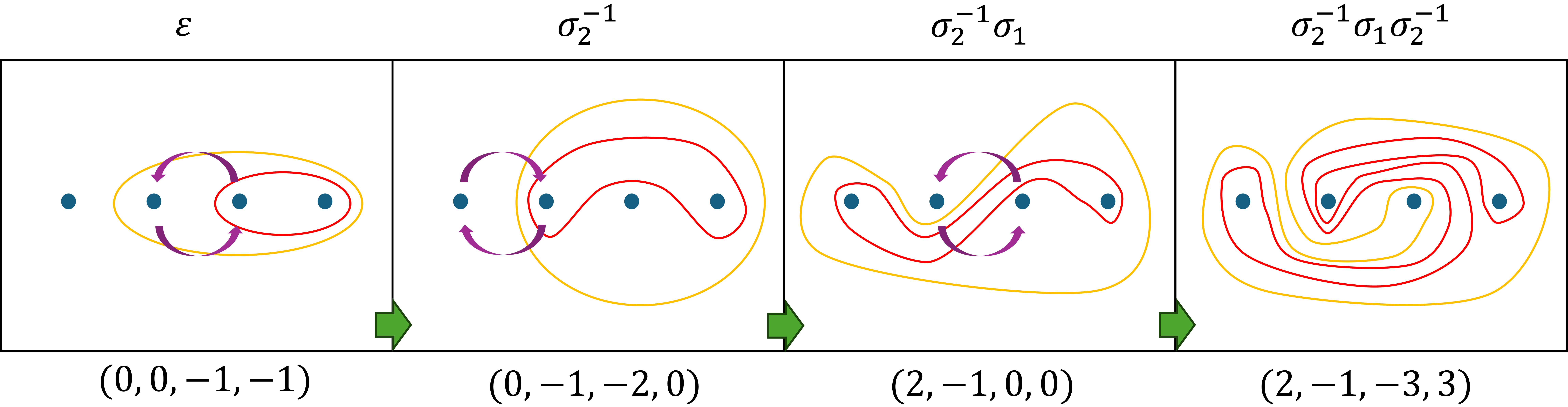}
    \caption{Action of braid $\sigma_2^{-1}\sigma_1\sigma_2^{-1}$ to multicurve and corresponding Dynnikov coordinates.
    The top shows braids,
the middle shows multicurves after these braids acted on $u$,
and the bottom shows the corresponding Dynnikov coordinates.
The Dynnikov coordinates corresponding to $\sigma_2^{-1}\sigma_1\sigma_2^{-1}$ are $(2,-1,-3,3)$.}
    \label{fig:Dynnikov_example}
    \Description{Four figures of twisted multicurves}
\end{figure}

Figure~\ref{fig:Dynnikov_example} shows examples of the action of the braid group and Dynnikov coordinates when $n=3$.

Since the values of the coordinates can be large, multiple-precision arithmetic is necessary for its implementation. The logarithmic magnitude of the coordinates is $O(l)$, where $l$ is the length of the braid word, and computational cost for one update (\ref{eq:Dyn_update}) is proportional to the logarithmic magnitude of the coordinates. Therefore, the time complexity of comparing braid words of length $l$ is $O(l^2)$, which is the best complexity among known algorithms.
Furthermore, since this calculation only requires simple arithmetic operations, it can be highly efficient in practice~\citep{dehornoy2008efficient}. See also \S~\ref{subsubsec:optimization_results}.

\subsection{Revised Prioritized Planning}\label{subsec:revised prioritized planning}
In this and the following subsections, we revert to the notation used prior to the end of \S~\ref{subsec:reduction to obstacle-free case}. Thus, $n$ denotes the number of agents and $r$ denotes the number of obstacles.

We adopt {\it revised prioritized planning} (RPP)~\citep{vcap2015prioritized} for pathfinding.
As classical prioritized planning, we fix a predetermined priority order for the agents and proceed to plan their paths one by one, avoiding collisions with the paths of previously planned agents. In RPP, in addition, the start positions of agents that have not yet been planned are also avoided, ensuring the completeness of the planning process under certain assumptions.

\begin{algorithm}[!t]
\caption{Homotopy-Aware RPP}
\label{algorithm:HRPP}
\begin{algorithmic}[1]
\State{Input a graph $G=(V,E)$ and starts and goals $(s_1, g_1),\ldots,(s_n,g_n)$.}
\State{Input the desired number $K$ of solutions}
\State{$\mathit{plans}\leftarrow \{\emptyset\}$}
\For{$i=1,\ldots,n$}
\State{$\mathit{Open},\mathit{Closed}\leftarrow \emptyset, \emptyset$}
\State{Initialize a map $D$}
\For{$\mathit{plan}\in\mathit{plans}$}
\State{Construct timed graph $G_{\mathit{plan}}=\left(V_{\mathit{plan}},E_{\mathit{plan}}\right)$}\label{line:construct_graph}
\State{Insert $(\mathit{plan},s_{\mathit{plan}},\varepsilon)$ to $\mathit{Open}$}
\State{$D[\mathit{plan},s_{\mathit{plan}},\varepsilon] \leftarrow 0$}
\EndFor
\State{$\mathit{newplans}\leftarrow\emptyset$}

\While{$\mathit{Open}\neq \emptyset$}\label{line:whileloop}
\State{Pop $(\mathit{plan}, v, w)$ from $\mathit{Open}$}\label{line:pop}
\State{Insert $(\mathit{plan}, v, w)$ to $\mathit{Closed}$}
\If{$v=g_{\mathit{plan}}$}
\State{Reconstruct path $p$ of agent $i$}
\State{Insert $\mathit{plan}\cup p$ to $\mathit{newplans}$}
\If{$|\mathit{newplan}| \geq K$}\label{line:enough plan}
\State{\textbf{break}}
\EndIf
\EndIf
\For{edge $e\in E_{\mathit{plan}}$ from $v$}
\State{$v', w', d'\leftarrow \mathrm{target}(e), \mathrm{NextBraid}(\mathit{plan},e,w), D[\mathit{plan},v,w]+\mathrm{length}(e)$}\label{line:nextbraid}
\If{$(\mathit{plan},v',w')\notin \mathit{Open}\cup\mathit{Closed}$}
\State{$D[\mathit{plan},v',w']\leftarrow d'$}
\State{Insert $(\mathit{plan},v',w')$ to $\mathit{Open}$}
\ElsIf{$D[\mathit{plan},v',w']>d'$}
\State{$D[\mathit{plan},v',w']\leftarrow d'$}
\EndIf
\EndFor
\EndWhile
\State{$\mathit{plans}\leftarrow\mathit{newplans}$}\label{line:end}
\EndFor
\State{\Return{$\mathit{plans}$}}
\end{algorithmic}
\end{algorithm}

Our method differs from conventional RPPs in the following ways:
First, we maintain multiple plans for higher-priority agents with different homotopies while planning paths one by one. 
Planning for the $i$-th agent is performed on the homotopy-augmented graph for the first $i$ agents in $\mathcal{C}_i(\mathcal{D})$.
For multiple plans for already planned agents, the planning for the next agent is performed on different graphs, but they are done in parallel.

Algorithm~\ref{algorithm:HRPP} presents the pseudocode of our homotopy-aware version of RPP, which generates homotopically distinct solutions.
The notation $\varepsilon$ denotes the Dynnikov coordinates corresponding to the empty braid.
At line~\ref{line:construct_graph}, we construct a graph $G_{\mathit{plan}}$, which is created by adding a time dimension to $G$ and removing parts colliding with $plan$, the same as \citet{silver2005cooperative}.\footnote{If continuous time is considered, you can use the method in \citet{phillips2011sipp}.}
In this construction, we also remove $s_{i+1},\ldots,s_n$ from $G_{\mathit{plan}}$.
(Strictly speaking, we allow an agent to reach one of such positions if the position is its goal, since it cannot finish otherwise.)
In the following lines, $s_{\mathit{plan}}$ and $g_{\mathit{plan}}$ denote the vertices corresponding to the start, which is $s_i$ with time $0$, and the goal, which is $g_i$ with a late enough time, respectively.
A vertex of the homotopy-augmented graph for $\mathit{plan}$ is represented by a pair $(v,w)$ of a vertex $v$ of $G_{\mathit{plan}}$ and a braid $w \in B_{r+i}$.\footnote{Strictly speaking, the positions of the first $(i-1)$ agents are also contained in a vertex of the homotopy-augmented graph in $\mathcal{C}_{i}(\mathcal{D})$. However, since these positions are determined by $\mathit{plan}$, they are omitted.}
At line~\ref{line:pop}, the popped element $(\mathit{plan}, v, w)$ is selected as that with the minimum $\mathrm{cost}(\mathit{plan})+D[\mathit{plan}, v, w]+h(v,g_{\mathit{plan}})$, where $h$ is the heuristic function.

The function $\mathrm{NextBraid}(\mathit{plan},e,w)$ at line~\ref{line:nextbraid} returns the Dynnikov coordinates updated from $w$ after the agent $i$ moves along $e$ and the agents from $1$ to $i-1$ move in accordance with $\mathit{plan}$.
It is calculated as follows. First, 
as explained in \S~\ref{subsec:word construction}, the swaps of the order of the agents' $x$-coordinates by these moves are enumerated. Second, for each swap, the Dynnikov coordinates are updated by (\ref{eq:Dyn_update}) in order.
After calculating the next braid $w'$, as in the conventional A* algorithm, the target vertex $(v',w')$ is checked to see if it has been visited.
If so, it is added.
Otherwise, its distance is updated or nothing happens.

For simplicity, we omit the details on path reconstruction.

\begin{remark}\label{remark:k-SNHP}
While we adopt the RPP approach for scalability in this paper,
we can also solve the $K$-shortest non-homotopic path planning by A* searching on the homotopy-augmented graph in $\mathcal{C}_n(\mathcal{}D)$.
\end{remark}

\subsection{Completeness}\label{subsec:completeness}
As RPP in the classical case, we can prove the completeness of the algorithm under some specific condition for problem instances. Note that this does not mean that our algorithm is practical only for such cases.

Before the completeness proposition, we state a lemma used in the proof.
\begin{lemma}
For any $1\leq k\leq n$, let $F_k: P_n\rightarrow P_{n-1}$ be the projection forgetting the $k$-th point. For $1\leq i < j \leq n$, let
\begin{equation}\label{eq:def_a}
    a_{i,j}:=\sigma_{j-1}\sigma_{j-2}\cdots\sigma_{i+1}\sigma_i^2\sigma_{i+1}^{-1}\cdots\sigma_{j-2}^{-1}\sigma_{j-1}^{-1}
    =\sigma_{i}^{-1}\sigma_{i+1}^{-1}\cdots\sigma_{j-2}^{-1}\sigma_{j-1}^2\sigma_{j-2}\cdots\sigma_{i+1}\sigma_{i} \in P_n.
\end{equation}
Then, the kernel of $F_k$ is generated by $a_{1,k},a_{2,k},\ldots,a_{k-1,k}$ and $a_{k,k+1},a_{k,k+2},\ldots,a_{k,n}$.
\end{lemma}
\begin{proof}
This is true for $k=n$~\citep{rolfsen2010tutorial}.
For arbitrary $k$, let $b_k:=\sigma_{n-1}\sigma_{n-2}\cdots\sigma_{k}\in B_n$.
Since $P_n$ is a normal subgroup, the conjunction by $b_k$ maps $P_n$ to $P_n$. Furthermore, $F_k(\alpha)=F_n(b_k\alpha b_k^{-1})$. Thus, the kernel of $F_k$ is generated by $b_k^{-1} a_{1,n} b_k,\ldots,b_k^{-1} a_{n-1,n} b_k$, which can be calculated as
\begin{equation}
    b_k^{-1}a_{i,n}b_k =
    \begin{cases}
    a_{i,k} & \text{if $i<k$} \\
    a_{k,i+1} & \text{if $i\geq k$}.
    \end{cases}
\end{equation}
\end{proof}

Intuitively, the following proposition says that our algorithm is complete when the obstacles and start and goal positions are separated enough with respect to the size of the agents.

\begin{proposition}\label{prop:completeness}
We assume the following conditions:
\begin{itemize}
\item An agent moving along the boundary of $\mathcal{D}$ never collides with obstacles or other agents staying at $s_1,\ldots, s_n$, $g_1,\ldots, g_n$, or $s_{i+1},\ldots,s_n$.
\item For any $1\leq i \leq n$, there exists a loop around $g_i$ in $\mathcal{D}$ such that it does not enclose any obstacle, other goal position, or $s_j$ with $j>i$, and an agent moving along it never collides with obstacles or other agents staying at any goal position or $s_j$ with $j>i$.
\item For any $1\leq i \leq n$, there exists a loop around $s_i$ such that it does not enclose any obstacle, other start position, or $g_j$ with $j < i$, and an agent moving along it never collides with obstacles or other agents staying at any start position or $g_j$ with $j < i$.
\end{itemize}
Then, when the roadmap $G$ is dense enough, for any homotopy class of solutions, Algorithm~\ref{algorithm:HRPP} provides a solution belonging to it for sufficiently large $K\gg0$.
\end{proposition}
\begin{proof}
We consider a problem instance $(s_1,g_1),\ldots,(s_n,g_n)$ satisfying the assumption and a braid $\alpha\in B_{r+n}$ corresponding to one of its solutions. For $0\leq i \leq n$, let $\alpha_i\in B_{r+i}$ be the braid corresponding to the paths for the first $i$ agents in $\alpha$.
We prove that our method generates a solution $p_i$ in the homotopy class corresponding to $\alpha_i$ for sufficiently large $K\gg 0$ by induction for $i$. When $i=0$, this is clear because no agent is considered, and $\alpha_0$ is the unit.

For $i\geq 1$, we assume that the plan $p_{i-1}$ with the homotopy corresponding to $\alpha_{i-1}$ is contained in $\mathit{plans}$ in Algorithm~\ref{algorithm:HRPP}.
Any reachable vertex in the homotopy-augmented graph can be found by repeating loop of Line~\ref{line:whileloop} a sufficient number of times, since the number of vertices in the homotopy-augmented graph at distances less a given value is finite.
Therefore, it suffices to show that there exists a plan $p_i$ with the homotopy corresponding to $\alpha_i$ on the basis of $p_{i-1}$.

Under the assumption of the problem instance, the $i$-th agent can reach its goal after agents from $1$ to $i-1$ have finished while avoiding the goal positions $g_1,\ldots,g_{i-1}$ and the start positions $s_{i+1},\ldots,s_n$. Therefore, there exists a plan $p_i'$ for the first $i$ agents on the basis of $p_{i-1}$. Let $\beta_i\in B_{r+i}$ be the braid corresponding to $p_i'$.

To construct the desired plan, we need to add the moves of the $i$-th agent after $p_i'$.
These moves must correspond to the braid $\beta_i^{-1}\alpha_i \in P_{r+i}$.
Let $c_1,\ldots,c_{r+i}$ be $o_1,\ldots,o_k, g_1,\ldots,g_i$ sorted by their $x$-coordinates, and suppose $g_i=c_k$ ($1\leq k \leq r+i$).
Since $\alpha_i$ and $\beta_i$ coincide when the $i$-th agent is forgotten, $\beta_i^{-1}\alpha_i$ is in the kernel of $F_{k}$, which is generated by $a_{1,k},\ldots,a_{k-1,k},a_{k,k+1},\ldots,a_{k,r+i}$ as previously proved.

Thus, it is sufficient to demonstrate that, after all $i$ agents have reached their goals, the $i$-th agent can move to construct any of $a_{1,k},\ldots,a_{k-1,k},a_{k,k+1},\ldots,a_{k,r+i}$ and their inverses, while avoiding $g_1,\ldots,g_{i-1}$ and $s_{i+1},\ldots,s_n$.

For $1\leq j<k$, the agent first moves up (in the positive direction of the $y$-axis) avoiding obstacles, $g_1,\ldots,g_{i-1}$, and $s_{i+1},\ldots,s_n$ by moving partially along the boundaries or the assumed loops. When avoiding the obstacle $O_l$, if the $x$ coordinate of $o_l$ is smaller than that of $c_k$, it avoids to the right (the positive side of the $x$-axis); otherwise, it avoids to the left. This is the same for $g_1,\ldots,g_{i-1}$.
Then, the agent moves to the left until it reaches the same $x$-coordinate as $c_j$ along the outside boundary and moves down to $c_j$ avoiding obstacles and other start or goal positions. When going down as well as going up, the agents avoid obstacles or goal positions represented by $c_l$ to the right (from our point of view) if the $x$-coordinate of $c_l$ is smaller than that of $c_j$; otherwise, it avoids to the left.
It circumvents $c_j$ counterclockwise or clockwise and returns to $c_k$ using the same path.
The braid for such a loop is given by:
\begin{equation}
\sigma_{k-1}\cdots\sigma_{j+1}\sigma_j^{\pm 2}\sigma_{j+1}^{-1}\cdots\sigma_{k-1}^{-1}=a_{j,k}^{\pm 1},
\end{equation}
where the sign depends on the direction for moving around $c_j$.
Similarly, for $k<j\leq n$,
the agent follows the path described above in reverse order to $c_j$, around $c_j$, and back. The braid is given by:
\begin{equation}
\sigma_{k}^{-1}\ldots\sigma_{j-2}^{-1}\sigma_{j-1}^{\pm 2}\sigma_{j-2}\ldots\sigma_{k}=a_{k,j}^{\pm 1}.
\end{equation}
Figure~\ref{fig:paths} illustrates examples of such paths.

When the roadmap is dense enough, these paths can be approximated by paths on the graph.
\begin{figure}[t]
\centering
    \begin{subfigure}{0.45\columnwidth}
        \includegraphics[width=\textwidth]{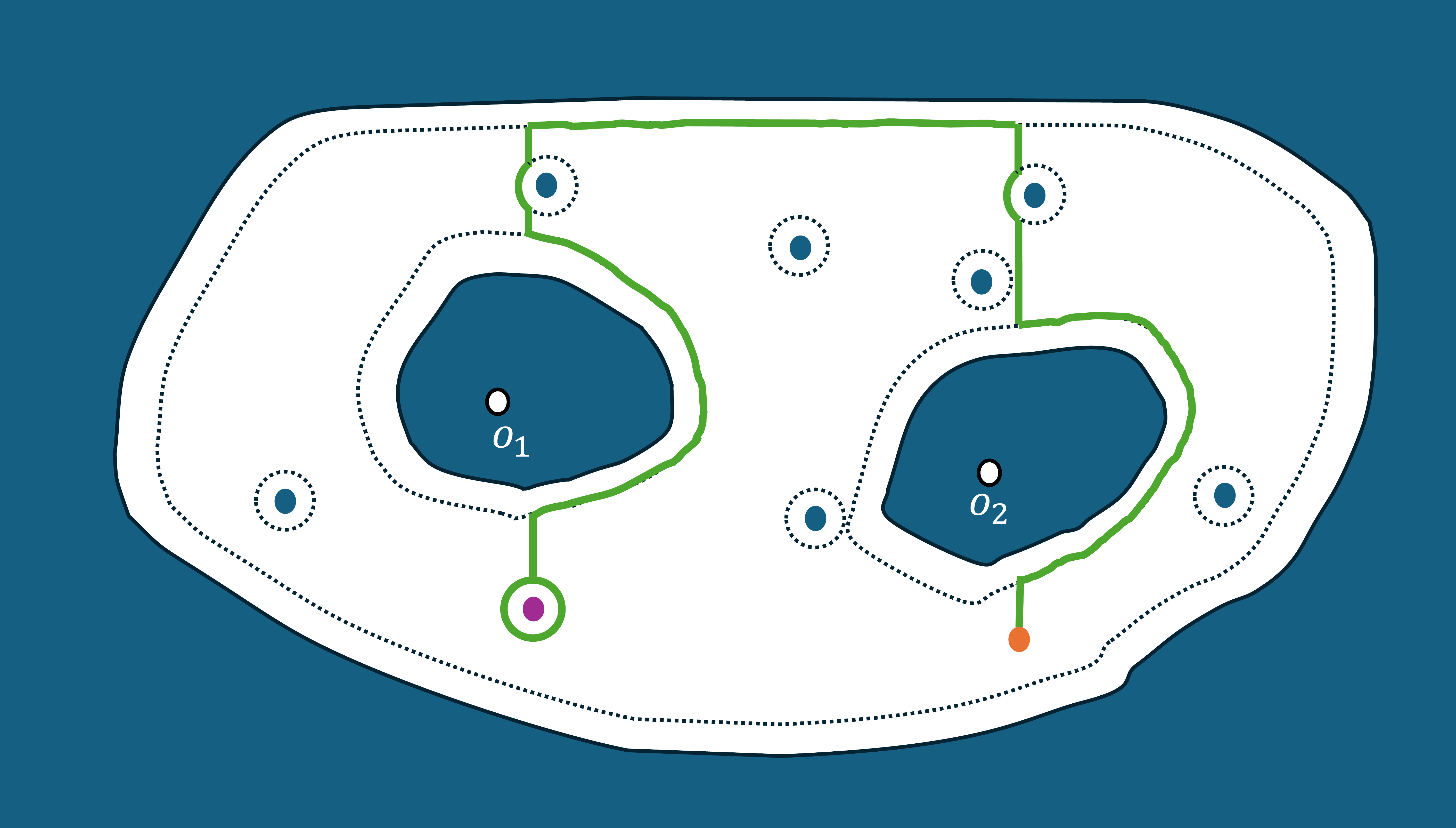}
        \caption{$j<k$}
    \end{subfigure}    
    \hspace{10pt}
    \begin{subfigure}{0.45\columnwidth}
        \includegraphics[width=\textwidth]{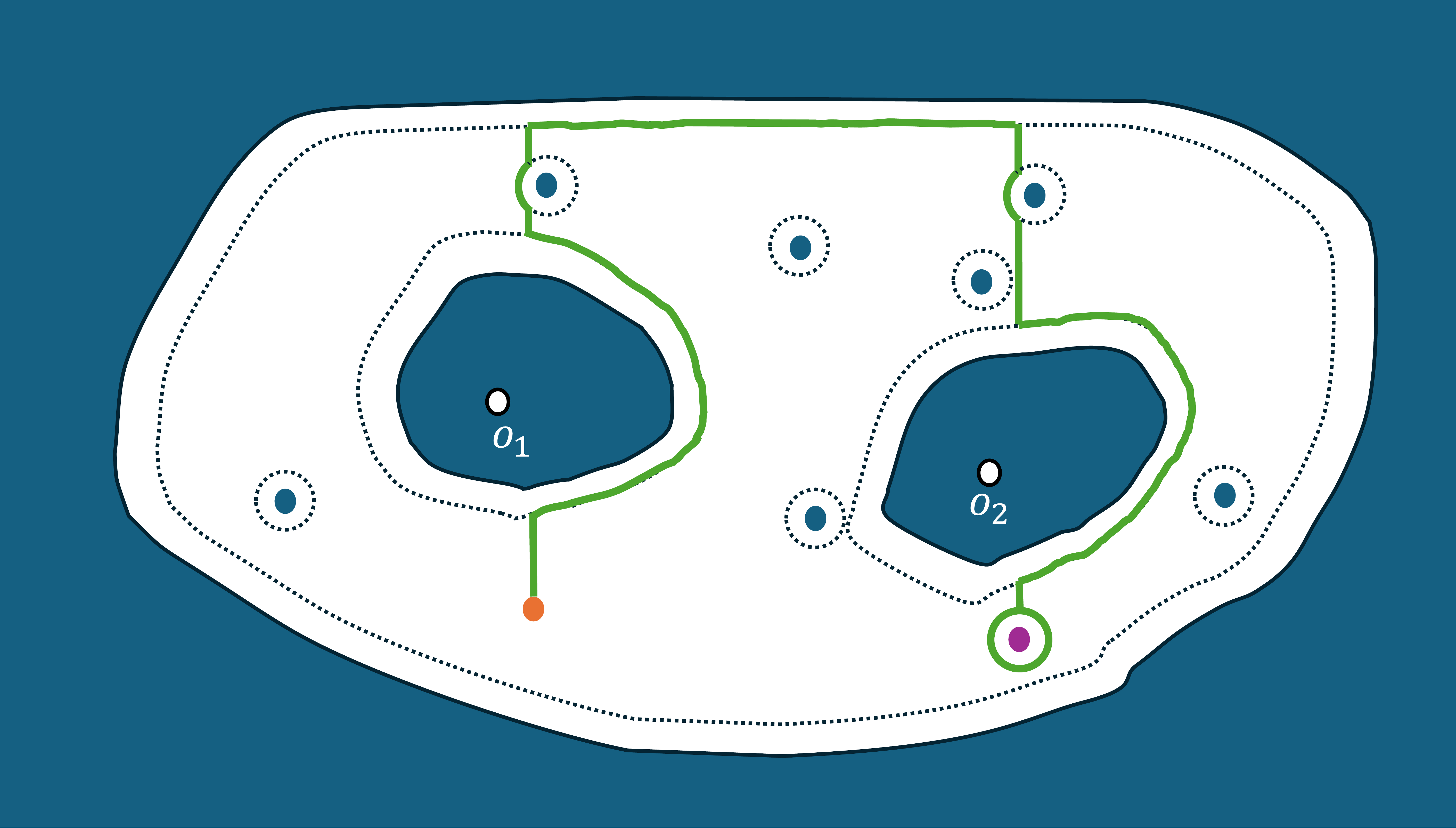}
        \caption{$k<j$}
    \end{subfigure}
\caption{Paths with homotopy corresponding to $a_{j,k}, a_{k,j}$ or its inverses. The circle enclosed by the path means $g_i=c_j$ and the circle at the endpoint of the path means $c_k$. Other circles mean other $c_1,\ldots,c_i$ or $s_{i+1},\ldots,s_n$.}
\Description{Two figures of a field with a path}
\label{fig:paths}
\end{figure}
\end{proof}

In particular, in the case of the grid, by the construction of paths in the proof, the following corollary holds.
To facilitate the word construction method in grid environments, we introduce a virtual slightly inclined $x$-axis, such that for any two grid cells $(i,j)$ and $(k,l)$, cell $(i,j)$ has a smaller x-coordinate than cell $(k,l)$ if and only if $i<k$ or $i=k, j<l$.

\begin{corollary}\label{prop:grid}
We consider cases where $\mathcal{D}$ is a region composed of grids, and $G$ is a four-connective grid graph.
We assume the following conditions:
\begin{itemize}
\item No start grid or goal grid is adjacent to obstacle grids vertically, horizontally, or diagonally.
\item No two start grids or two goal grids are adjacent vertically, horizontally, or diagonally.
\item For any $1\leq i < j \leq n$, $g_i$ and $s_j$ are not adjacent vertically, horizontally, or diagonally.
\end{itemize}
Then, for any homotopy class of solutions, Algorithm~\ref{algorithm:HRPP} generates a solution belonging to it for sufficiently large $K \gg 0$.
\end{corollary}

\begin{remark}
Proposition~\ref{prop:completeness} states that, if the assumptions are met, then no homotopy classes are excluded in principle. However, the value of $K$ required to generate a solution belonging to a specific homotopy class may be large.
To find a solution with the homotopy corresponding to a given braid, Algorithm~\ref{algorithm:HRPP} can be modified as follows.
Let $\alpha \in B_{r+n}$ be the given braid and let $\alpha_i \in B_{r+i}$ be the image of $\alpha$ under the projection forgetting all agents after the $i$-th one.
Unlike Algorithm~\ref{algorithm:HRPP}, we search for only one plan when planning for each agent, as in conventional RPP.
For the $i$-th agent, we search for a path to the vertex $(g_{\mathit{plan}}, \alpha_i)$ in the homotopy-augmented graph.
From the proof of Proposition~\ref{prop:completeness}, it is derived that, under the same assumptions as in the Proposition~\ref{prop:completeness}, this algorithm can find a desired solution.
\end{remark}

\section{Experiments}
\begin{figure*}[t]
\centering
\begin{minipage}{0.74\textwidth}
\begin{subfigure}{0.30\textwidth}
    \includegraphics[width=\columnwidth]{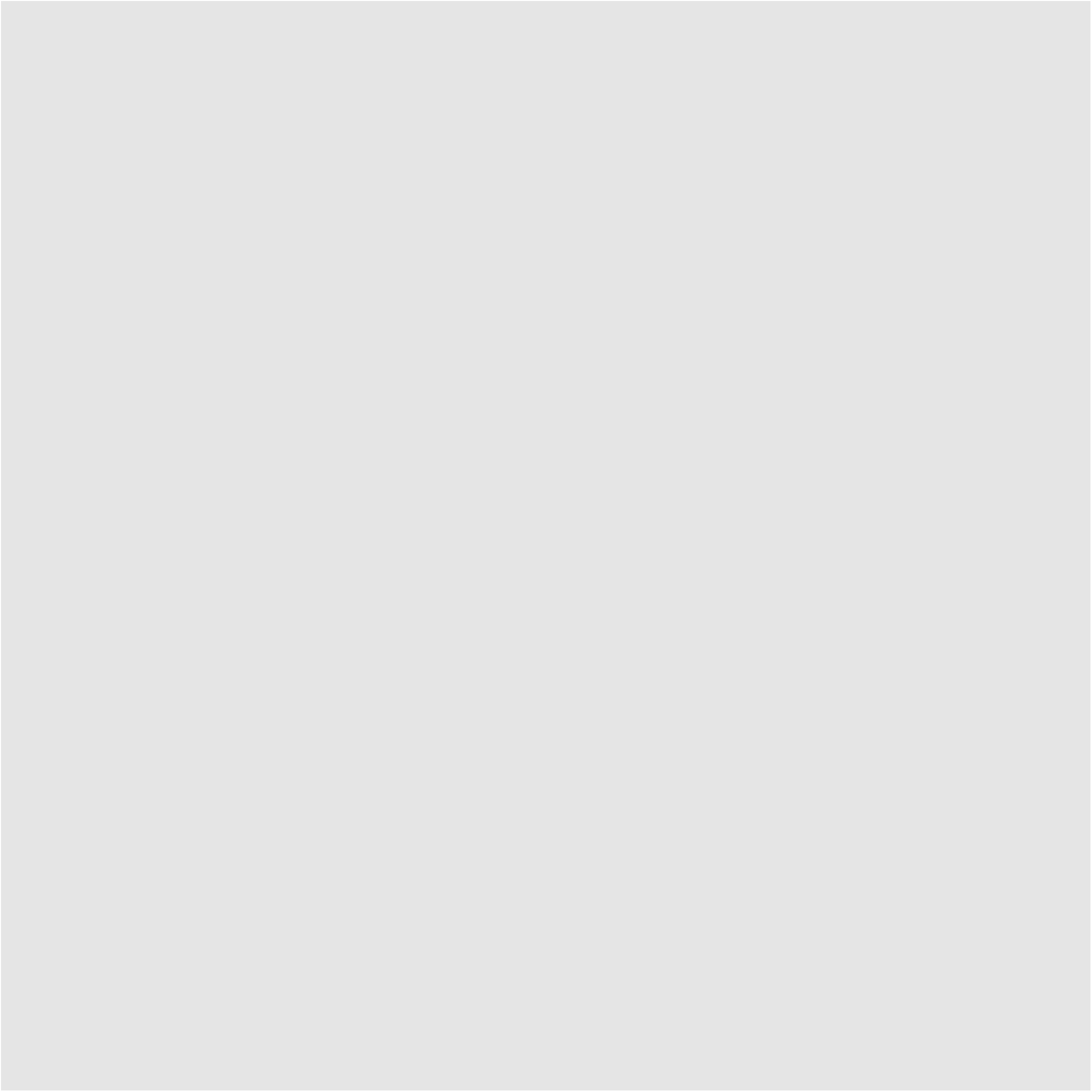}
\caption{\texttt{empty-48-48}}
\end{subfigure}
\begin{subfigure}{0.30\textwidth}
    \includegraphics[width=\columnwidth]{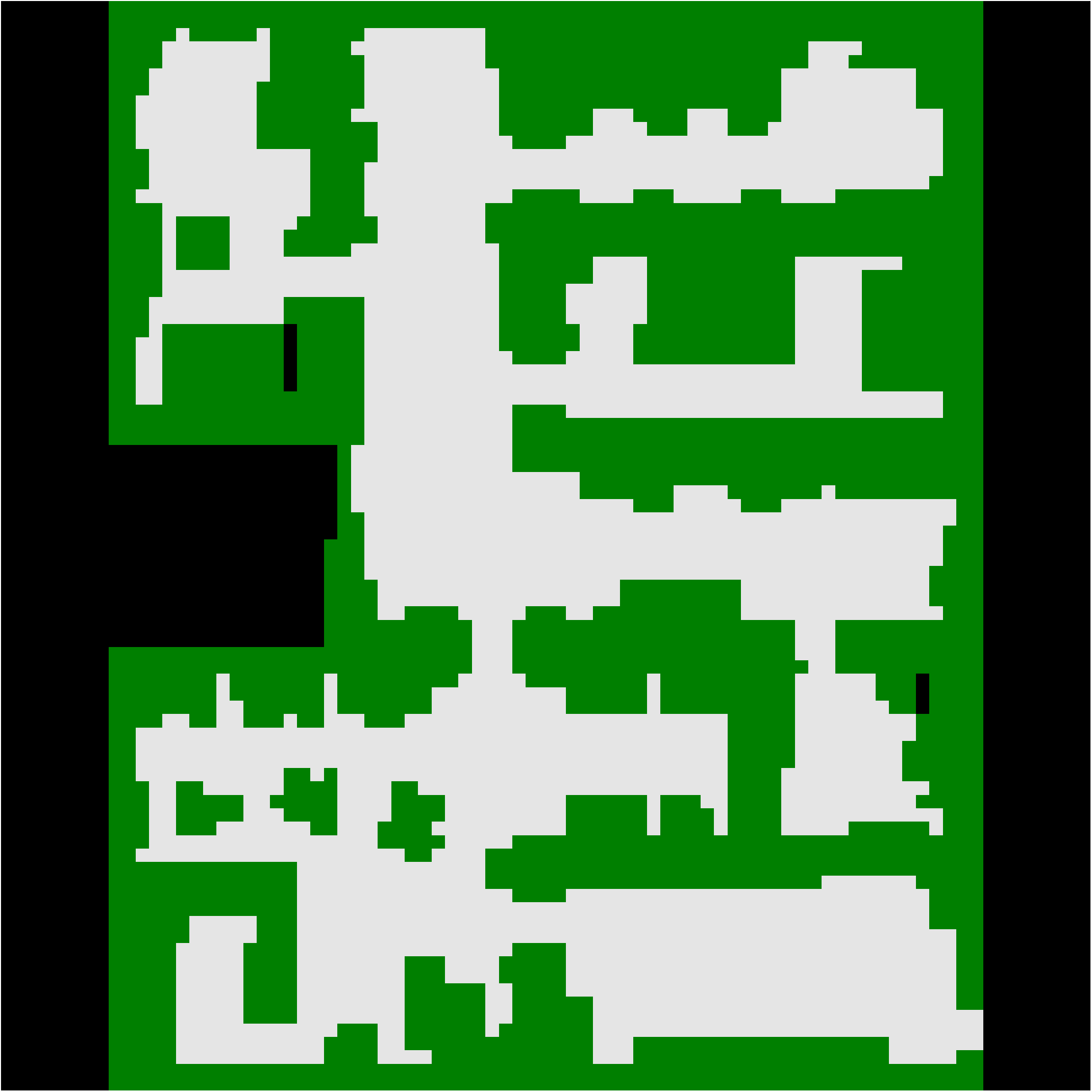}
\caption{\texttt{den312d}}
\end{subfigure}
\begin{subfigure}{0.30\textwidth}
    \includegraphics[width=\columnwidth]{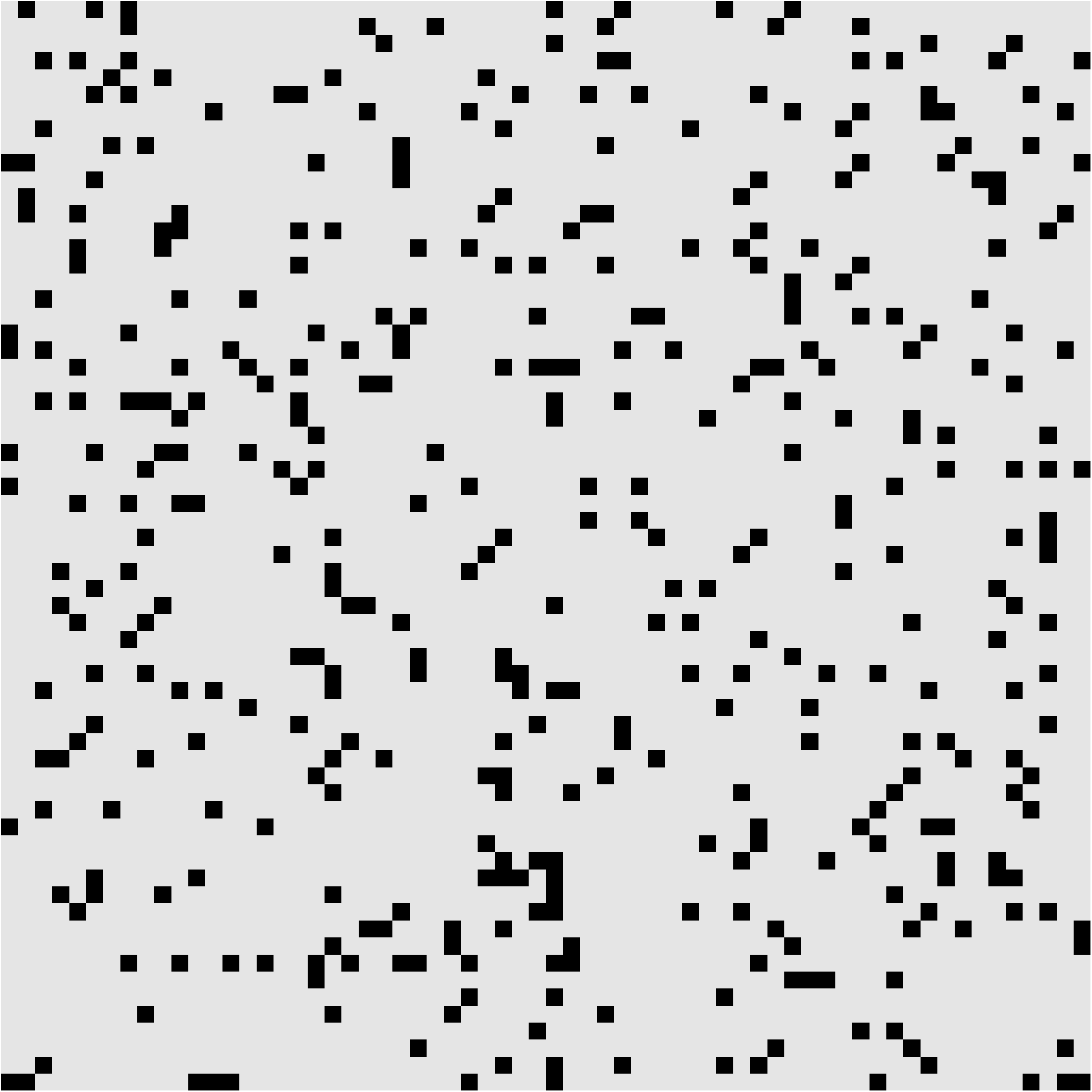}
\caption{\texttt{random-64-64-10}}
\end{subfigure}
\caption{Maps from Moving AI Benchmark for evaluation of runtime.}
\label{fig:maps}
\Description{The same as in movingai.com}
\end{minipage}
\begin{minipage}{0.24\textwidth}
    \includegraphics[width=\columnwidth]{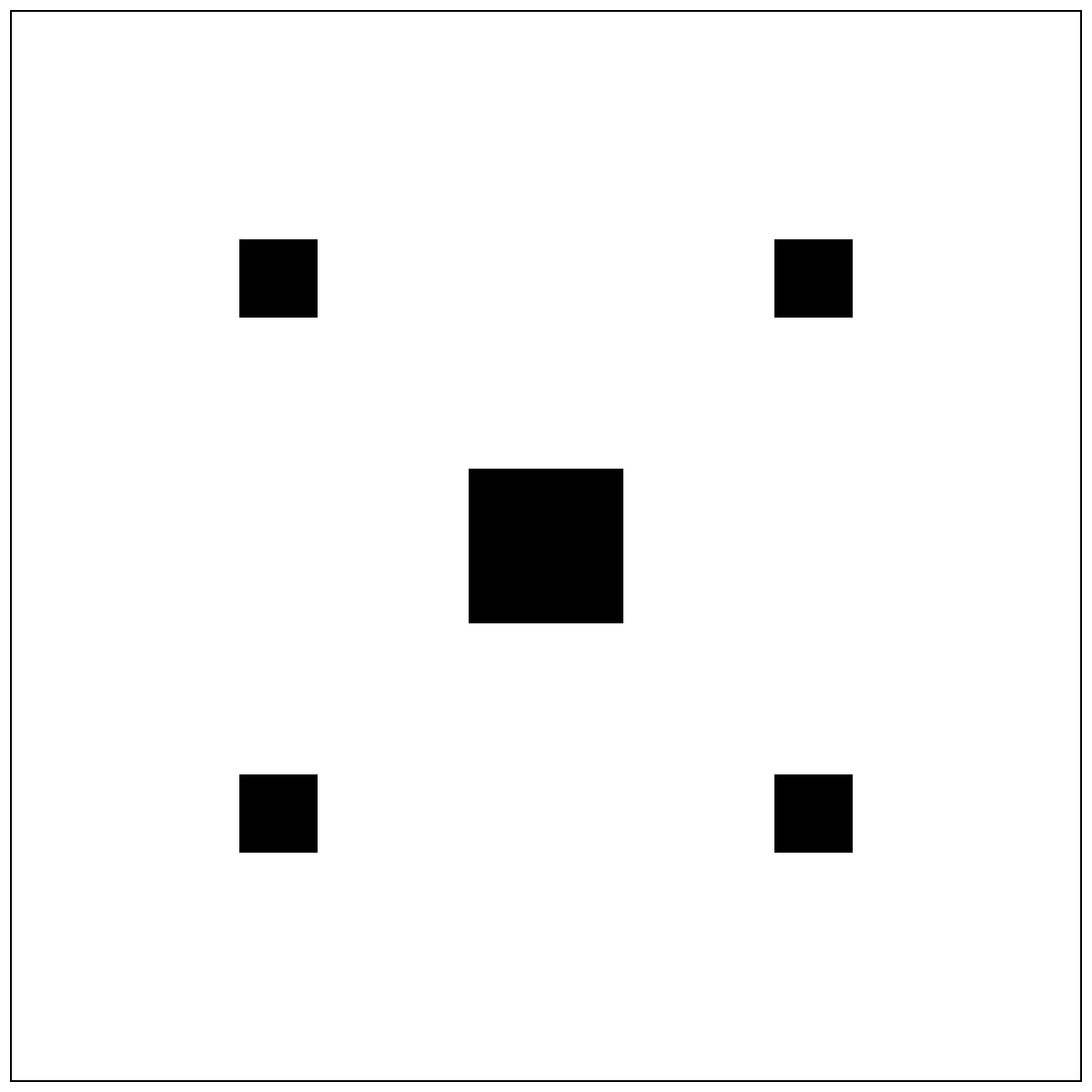}
\caption{Grid map with five obstacles for optimization experiment.}
\label{fig:obstacle-14-14}
\end{minipage}
\Description{One large square obstacle in center, four small square obstacles between the center and each corners}
\end{figure*}

We conducted two experiments. In the first experiment, we measured the runtime of our method to assess its scalability, comparing it with that of a method using the Dehornoy order and the handle reduction algorithm~\citep{dehornoy1997fast} instead of Dynnikov coordinates.
In the second experiment, our focus was to demonstrate the effectiveness of generating homotopically distinct coarse solutions for planning low-cost trajectories. We demonstrate how our approach can lead to improved trajectories by comparing the results obtained after optimization.
The codes used for these experiments are available at \url{https://github.com/omron-sinicx/homotopy-aware-MAPP}.

\subsection{Implementation}
In our implementation, we precomputed the minimum distance to $g_{\mathit{plan}}$ for all vertices of $G_{\mathit{plan}}$ and adopted it as the heuristic function $h$.
Our implementation was conducted in C++, and we used the priority queue and map data structures from the standard library to manage $\mathit{Open}$, $\mathit{Closed}$, and $D$.
Dynnikov coordinates were implemented by using the GMP library.

We used Intel(R) Xeon(R) Gold 6338 CPUs @ 2.00 GHz for these experiments.

\subsection{Evaluation of Runtime}\label{subsec:eval_rantime}


We measured runtime of our method (\textbf{Dyn}) with the following two settings:

\begin{itemize}
\item[(A)] To observe the increase in runtime with the number of agents, we ran our algorithm with $n=500$ and $K=100$, where $n$ is the number of agents and $K$ is the number of solutions to find, measuring the runtime from the beginning after each agent's planning process was completed (line~\ref{line:end} in Algorithm~\ref{algorithm:HRPP}).
\item[(B)] To observe the increase in runtime with the value of $K$, we ran our algorithm with $n=30$ and $K=1,2,3,10,3,100$.
\end{itemize}

\subsubsection{Problem Instances}
We evaluated the runtime performance using random instances of the multi-agent pathfinding problem.
We used the three grid maps, \texttt{empty-48-48}, \texttt{den312d}, and \texttt{random-64-64-10}, illustrated in Figure~\ref{fig:maps}, from Moving AI MAPF Benchmarks~\citep{stern2019multi} as environments.
The sizes of the maps were $48\times48$, $81\times65$, and $64\times64$, respectively.
The numbers of the connected components of obstacles ($r$ in the previous section) were $0$, $4$, and $241$, respectively.
For each of the experiments (A) and (B), 
we generated $10$ instances on each map.
To ensure at least some solutions can be found, we imposed a condition on problem instances that guaranteed that the classical RPP would find a solution, i.e., for any $1\leq i < n$, the agent $i$ can reach the goal without visiting $s_{i+1},\ldots,s_n$ or $g_1\ldots, g_{i-1}$.

\subsubsection{Baseline}
We also recorded the maximum absolute value of the Dynnikov coordinates for the first generated solution at the same time.
For comparison, we also ran a method (\textbf{HR}) that is the same with ours except that it uses the Dehornoy order~\citep{dehornoy1994braid} to maintain braids and the handle-reduction algorithm~\citep{dehornoy1997fast} to determine the order.
This algorithm has long been known as a practically efficient method to solve the word problem for braids, while its theoretical complexity remains unknown~\citep{dehornoy2008ordering,dehornoy2008efficient}.
More precisely, we used the ``FullHRed" algorithm as a specific strategy for the handle reduction method.

\subsubsection{Results and Discussion}

\begin{figure*}[t]
    \includegraphics[width=\columnwidth]{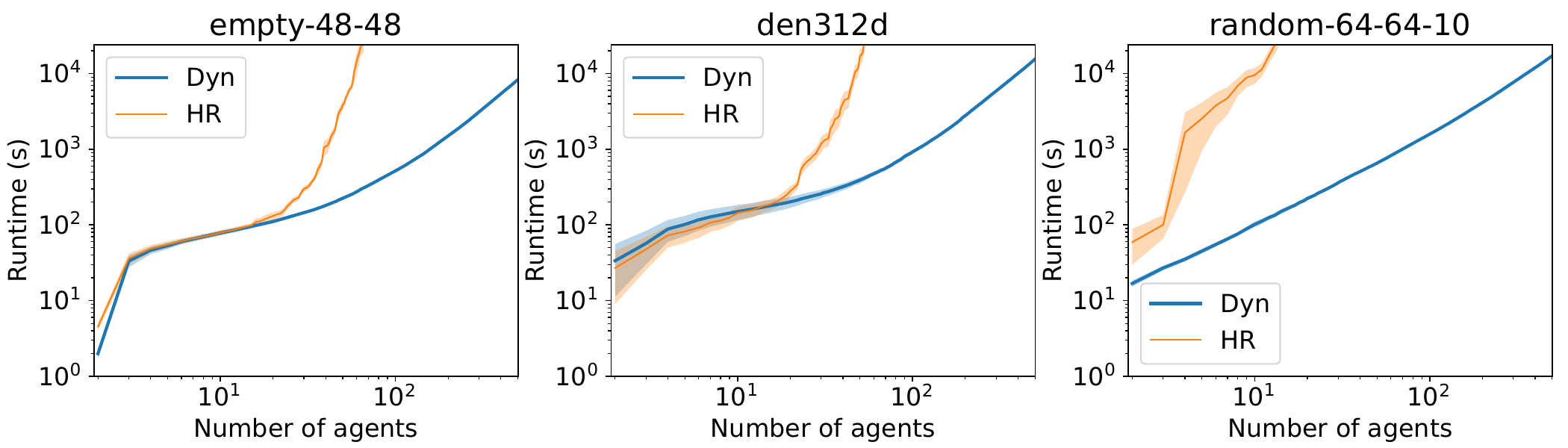}
    \caption{Log-log plots of runtime with respect to the number of agents for our method and the baseline method using the Dehornoy order. Colored areas represent standard errors.}
    \label{fig:log-log plots}
    \Description{Dyn wins. Runtime increased more rapidly in HR that in Dyn.}
\end{figure*}

\begin{figure*}[t]
\centering
\begin{minipage}{0.48\textwidth}
    \includegraphics[width=\columnwidth]{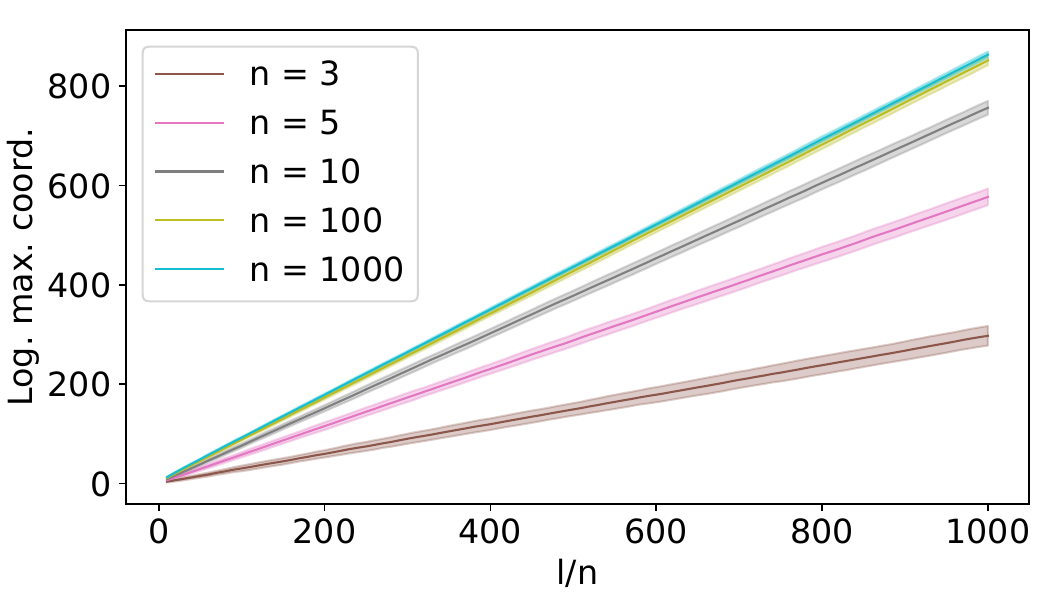}
\caption{Logarithms of maximum absolute values of coordinates for random braids with respect to lengths of braid words divided by the number $n$ of agents, for various values of $n$.  Colored areas represent standard derivations.}
\label{fig:random_braids}
\end{minipage}
\hspace{0.02\textwidth}
\begin{minipage}{0.48\textwidth}
    \includegraphics[width=\columnwidth]{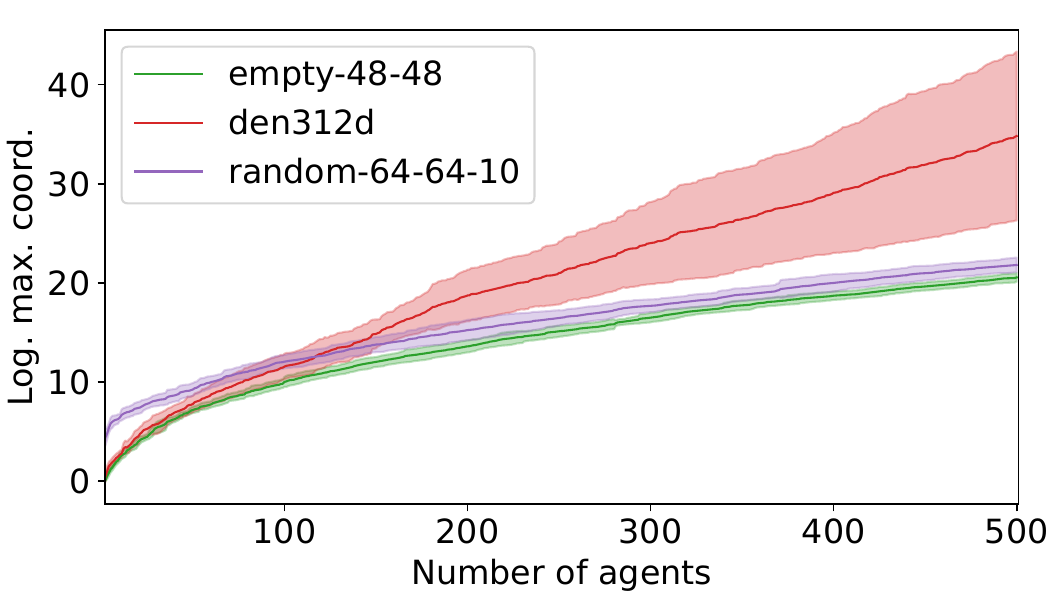}
\caption{Logarithms of maximum absolute values of coordinates in experiments with respect to the number of agents, for each environment map. Colored areas represent standard derivations.}
\label{fig:max_coord}
\Description{Results as discussed in the text}
\end{minipage}
\end{figure*}

Figure~\ref{fig:log-log plots} illustrates the results of the runtimes in the experiment (A), which show that the runtime order of with respect to the number $n$ of agents was smaller for our method than for \textbf{HR}.
Moreover, while the runtime increased significantly in the case with many obstacles (\texttt{random-64-64-10}) compared with the other cases in \textbf{HR}, it did not change much in our method.

In the case of the empty map, the runtime of \textbf{HR} was approximately $\Theta(n^5)$, which means that the runtime for the planning for the $n$-th agent was $\Theta(n^4)$.
This is expected when the handle reduction is the bottleneck because the lengths $l$ of braid words are expected to be roughly $\Theta(n^2)$ and the time complexity of the handle reduction method is conjectured to be $\Theta(l^2)$~\citep{dehornoy2008ordering}.

On the other hand, the runtime order of our method seemed to be roughly $\Theta(n^2)$, which was smaller than the expected order explained below for calculating Dynnikov coordinates, which means that this calculation was not a bottleneck at the scale of our experiment.
Indeed, when we profiled the performance of a single run by using gprof, comparison and update of Dynnikov coordinates accounted for $7.2\%$ and $12.3\%$ of runtime, respectively, even in \texttt{random-64-64-10}.

As mentioned before, the complexity of one update calculation (\ref{eq:Dyn_update}) was $\Theta(\log X)$, where $X$ is the absolute value of the coordinates. While $\log X$ is bounded as $O(l)$, where $l$ is a length of the braid word, for arbitrary words, it is expected to be $\Theta(l/n)$, where $n$ is the number of strands, for random words.
To confirm this, we generated $100$ random braid words with a length of $1000n$ for each value of $n=3,5,10,100,1000$ and calculated the maximum absolute values of the Dynnikov coordinates. The results shown in Figure~\ref{fig:random_braids} demonstrate that the expected value of $\log X$ is linear to $l/n$ and that the coefficient is independent of $n$ when $n$ is large enough. 

In our method, since the expected lengths of braid words for the planning of the $n$-th agent is roughly $\Theta(n^2)$, $\log X$ is expected to be $\Theta(n)$, which is consistent with the plot of the logarithms of the maximum absolute values of coordinates with respect to $n$ as shown in Figure~\ref{fig:max_coord} when $n$ is enough large.
On the other hand, this figure also shows that the coefficients vary depending on the shapes of maps and are larger for complex maps such as \texttt{den312d}.
The number of update calculations for the planning for the $n$-th agent is roughly estimated to be $\Theta(n^2)$, so the runtime for the planning for the $n$-th agent is expected to be $\Theta(n^3)$.
Thus, the total runtime for calculating Dynnikov coordinates for the planning of $n$ agents is expected to be $\Theta(n^4)$. When the number of agents is much larger, this cost will dominate.

\begin{figure*}[t]
    \includegraphics[width=\columnwidth]{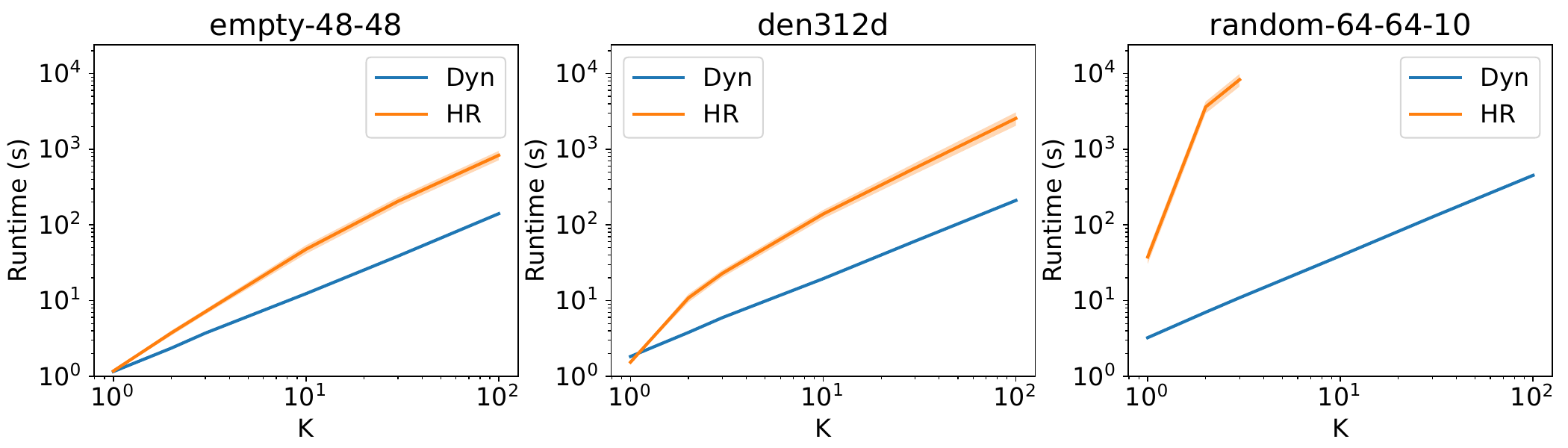}
    \caption{Log-log plots of runtime with respect to the number $K$ of solutions to find for our method and the baseline method using the Dehornoy order. Colored areas represent standard errors.}
    \label{fig:log-log plots with K}
    \Description{In Dyn, runtime increased linearly. In HR, runtime increased more rapidly.}
\end{figure*}
\begin{figure*}[t]
    \includegraphics[width=\columnwidth]{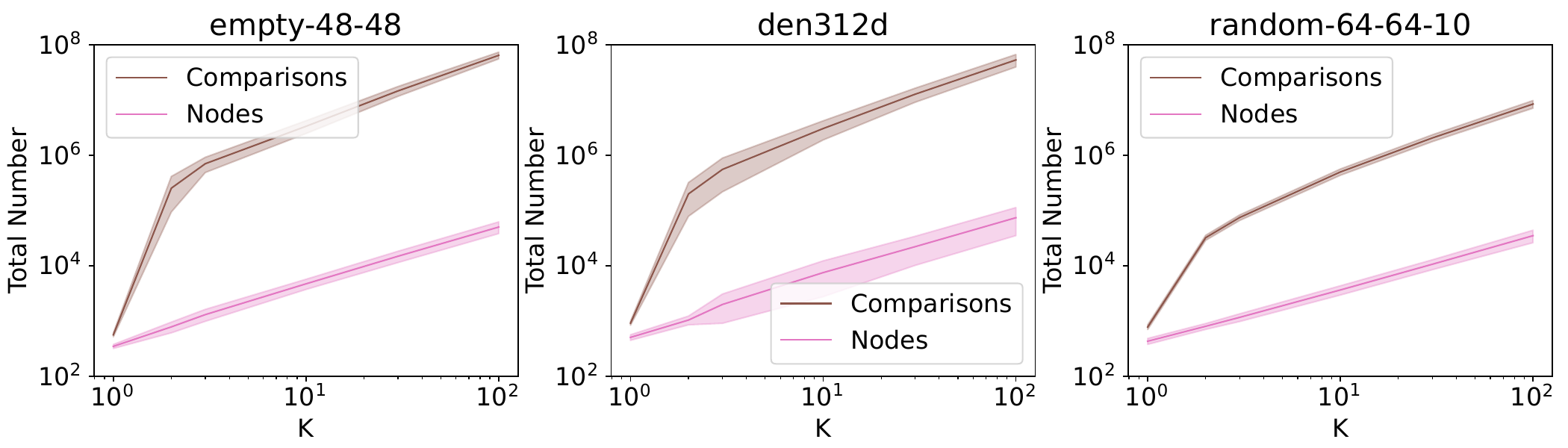}
    \caption{Log-log plots of the total numbers of braid comparisons and created nodes with respect to the number $K$ of solutions to find. Colored areas represent standard deviations.}
    \label{fig:counts with K}
    \Description{Number of nodes increased linearly. Number of comparisons increased more rapidly.}
\end{figure*}

Figure~\ref{fig:log-log plots with K} presents the runtime results from experiment (B). The runtime in \textbf{Dyn} increased approximately linearly with respect to $K$, whereas in \textbf{HR} the increase was slightly steeper than linear. 
The difference in runtime slope appears to stem from the fact that braid comparisons constituted the bottleneck in \textbf{HR}, but not in \textbf{Dyn}.
Figure~\ref{fig:counts with K} illustrates the number of braid comparisons performed and nodes created during algorithm execution\footnote{Strictly speaking, this is the result in \textbf{Dyn}, so the number of comparisons is not exactly the same in \textbf{HR}, but there is no essential difference.}.
Indeed, the number of nodes grew roughly linearly, while the number of comparisons exhibited a slope somewhat steeper than linear.
Since braid comparisons are performed only for nodes on the same timed graph $G_{\mathit{plan}}$, and the number of such timed graphs is $K$, one would expect the number of comparisons to grow linearly with $K$ if the searches were distributed evenly across the graphs. The deviation from linearity therefore suggests that the exploration was uneven.

The results of experiments (A) and (B) confirmed that our method \textbf{Dyn} is clearly superior to the baseline method \textbf{HR} in scalability with respect to the number of agents, and also outperforms it in scalability with respect to the number of solutions $K$ to be found.

\subsection{Optimization Experiment}
To confirm that generating multiple homotopically distinct solutions is useful for solving practical problems, we conducted an experiment in this subsection.

We considered the multi-agent trajectory optimization problem in a continuous planar domain and used the following strategy to solve it.
\begin{enumerate}
\item We generate multiple solutions on grids.
\item We continuously optimize the found solutions and select the best optimized trajectories.
\end{enumerate}
To generate initial solutions, we ran our algorithm to produce homotopically distinct solutions, and compared them with baseline methods.

\subsubsection{Problem Setting}
We consider multi-agent trajectory optimization in a continuous planar domain, where the objective function is based on acceleration~\citep{zucker2013chomp}.
This objective function can be justified from a practical point of view because the energy required for an agent to move and the force felt by something inside the agent depend on its acceleration.

Let $n$ be the number of agents.
Let $\mathcal{D}\subseteq \mathbb{R}^2$ be the domain in which the agents can move.
As in \S~\ref{subsec:problem setting},
we assume that $\mathcal{D}$ has the form $D_0\setminus(O_1\cup \cdots \cup O_r)$.
Let $r$ be the collision radius of the agents. Let $d_l$ be the distance function from $O_l$ for $l>0$ and that from the outside boundary for $l=0$.
The start positions $s_1,\ldots,s_n\in \mathcal{D}$ and goal positions $g_1,\ldots,g_n\in \mathcal{D}$ for all agents are given.

The goal is to find collision-free trajectories for all agents with the lowest sum of costs.
The cost is defined as the integral of the squared acceleration norm over time, which is regularized to the unit interval.
In summary, we want to find $\gamma_1, \dots, \gamma_n: [0,1] \rightarrow \mathcal{D}$ that minimize the following cost
\begin{equation}
    \frac{1}{2}\sum_{i=1}^n\int_{0}^1\left\|\frac{d^2\gamma_i}{dt^2}(t)\right\|^2dt
\end{equation}
under the following conditions:
\begin{equation}\label{eq:trajectory_constraint}
\gamma_i(0) = s_i, \, \gamma_i(1) = t_i, \, \frac{d\gamma_i}{dt}(0)=\frac{d\gamma_i}{dt}(1)=\mathbf{0},
\end{equation}
\begin{equation}\label{eq:agent_collision_condition}
\|\gamma_i(t)-\gamma_i(t)\| \geq 2r \text{ when } i \neq j,
\end{equation}
\begin{equation}\label{eq:obstacle_condition}
d_l(\gamma_i(t)) \geq r.
\end{equation}

\subsubsection{Problem Instances}
We used two small grid maps: an empty $14\times 14$ map and a handmade map of the same size with five obstacles as shown in Figure~\ref{fig:obstacle-14-14}.
We set $r=\sqrt{2}/4$, which is the maximum radius at which no collisions occur in grid map solutions.
For each map, we generated $100$ problem instances with start positions selected at random to be distinct and similarly selected goal positions.

\subsubsection{Compared Methods}
To solve the above problem instances, we first generated initial plans and then optimized them in the way described in the next subsubsection.
As a method to generate initial plans, we compared the following three approaches:
\begin{itemize}
\item \textbf{Ours}: For each instance, we generated $100$ homotopically distinct solutions with our method ($K=100$).
\item \textbf{Optimal one (OO)}: We generated one optimal solution on the grid for each instance.
\item \textbf{Revised prioritized planning with various priority orders (PPvP)}: We generated $100$ solutions for each instance by revised prioritized planning with randomly selected priority orders.
\end{itemize}

The intention of comparing our strategy to \textbf{OO} was to confirm that generating multiple initial plans is effective.
Also, we added \textbf{PPvP} to the comparison to answer the question of whether it is possible to generate diversified plans by simply generating solutions at random without homotopical consideration.

\subsubsection{Optimization Method}
After generating the initial solutions, we proceeded to optimize them.
To do this, we relax the collision-free conditions (\ref{eq:agent_collision_condition}) and (\ref{eq:obstacle_condition}) and incorporate them into the cost function~\citep{kasaura2023periodic}.
Let $\gamma_1, \ldots, \gamma_n$ be the trajectories of the agents. The function $C$ to be minimized is defined as
\begin{equation}
\begin{split}
C(\gamma_1,\ldots,\gamma_n):=&\frac{1}{2}\sum_{i=1}^n\int_{0}^1\left\|\frac{d^2\gamma_i}{dt^2}(t)\right\|^2dt 
    +c_{\mathrm{c}}\sum_{i\neq j}\int_0^1\max\left\{\frac{1}{\left\|\gamma_i(t)-\gamma_j(t)\right\|}-\frac{1}{2r},0\right\}^2dt\\
    &+c_{\mathrm{c}}\sum_{i=1}^n\sum_{l=0}^r\int_0^1\max\left\{\frac{1}{d_l(\gamma_i(t))}-\frac{1}{r},0\right\}^2dt,
\end{split}
\end{equation}
where $c_{\mathrm{c}}$ is a penalty coefficient.

More precisely, we approximate the trajectories of the $i$-th agent by $(L+1)$ timed waypoints $(p_{i,0},0), \allowbreak (p_{i,1},1/L), \allowbreak \ldots,(p_{i,L}, 1)$ with $p_{i,0}=s_i$ and $p_{i,L}=g_i$, where $L=100$. The cost function is then reformulated as
\begin{equation}
    \begin{split}
    C(p_{1,0},\ldots,p_{1,L},p_{2,0},\ldots,p_{n,L}):=&
    \frac{L^3}{2}\sum_i
    \sum_{k=0}^{L}\left\|p_{i,k+1}+p_{i,k-1}-2p_{i,k}\right\|^2\\
    &+\frac{c_{\mathrm{c}}}{L}\sum_{i\neq j}\sum_{k=0}^L\max\left\{\frac{1}{\left\|p_{i,k}-p_{j,k}\right\|}-\frac{1}{2r},0\right\}^2\\
    &+\frac{c_{\mathrm{c}}}{L}\sum_{i=1}^n\sum_{l=0}^r\sum_{k=0}^L\max\left\{\frac{1}{d_l(p_{i,k})}-\frac{1}{r},0\right\}^2,\\
    \end{split}
\end{equation}
where $p_{i,-1}:=p_{i,0}$ and $p_{i,L+1}:=p_{i,L}$.
We solved this continuous optimization problem with the Levenberg-Marquardt~\citep{levenberg1944method,marquardt1963algorithm} algorithm implemented in g$^2$o 1.0.0~\citep{kummerle2011g}.
The $p_{i,k}$ values were initialized by the initial plan on the grid. We optimized them with $10000$ steps of the Levenberg-Marquardt algorithm. The collision-penalty constant $c_c$ was initially set to $10^6$ and multiplied by $1.001$ after every optimization step.

\subsubsection{Results and Discussion}\label{subsubsec:optimization_results}
\begin{figure}[t]
    \centering
    \includegraphics[width=\columnwidth]{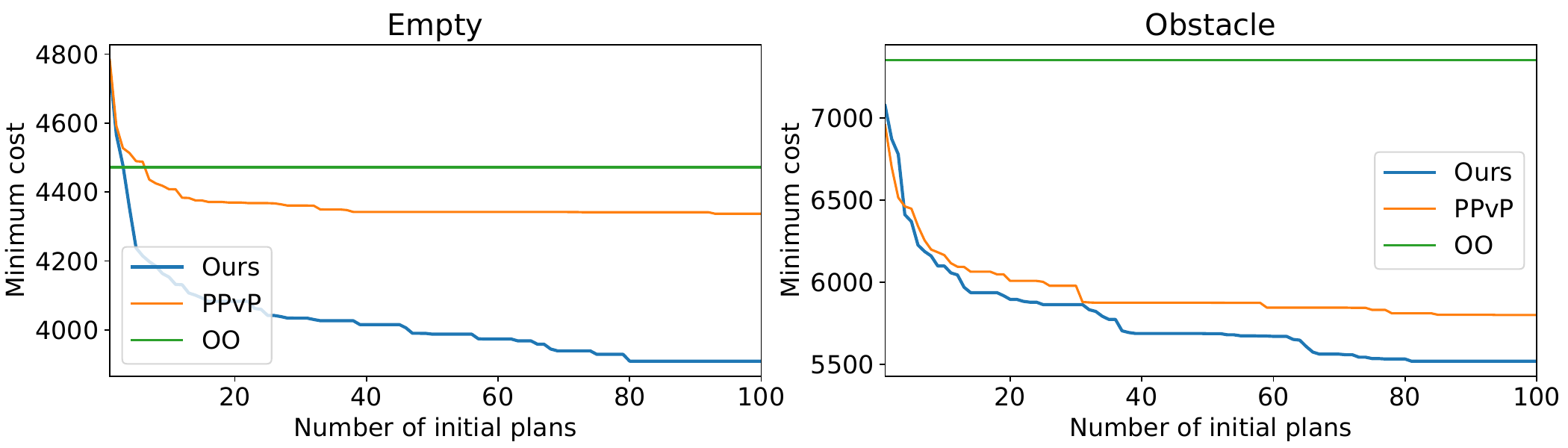}
    \caption{For \textbf{Ours} and \textbf{PPvP}, curves represent minimum cost of plans after optimization among first $N$ generated plans, averaged over 100 instances, with respect to $N$. Bar for \textbf{OO} represents the cost after optimization, averaged over 100 instances.}
    \label{fig:results}
    \Description{Our method has low costs, especially in the empty map}
\end{figure}
\begin{figure}[t]
    \centering
    \includegraphics[width=\columnwidth]{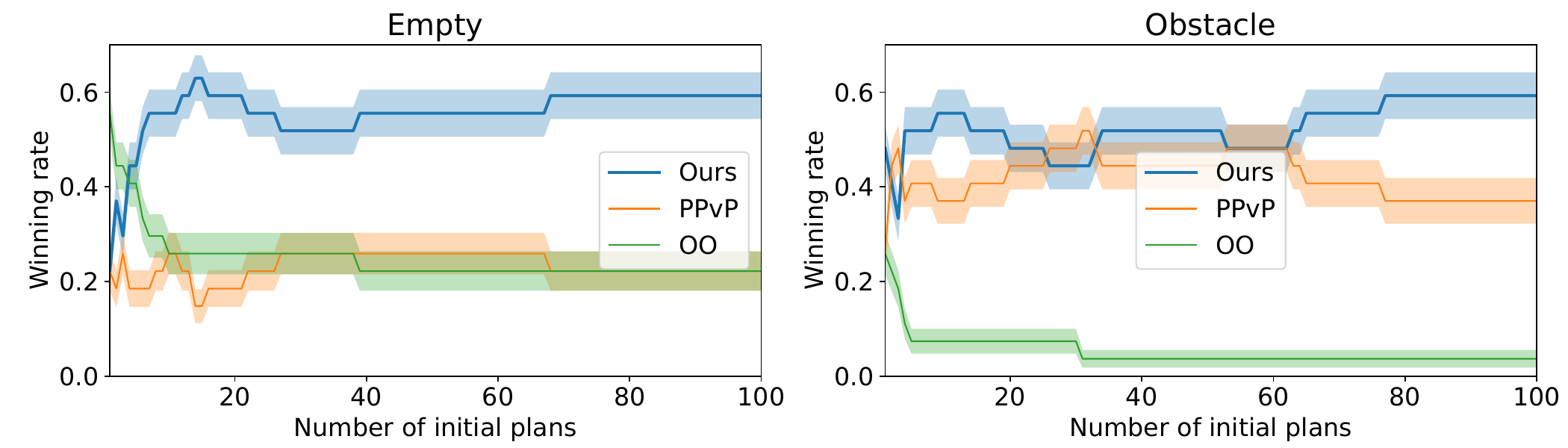}
    \caption{For each method, curves show rate of instances in which the best of the first $N$ plans generated by it wins over that of other methods, with respect to $N$. Colored areas represent standard errors.}
    \label{fig:wr_graph}
    \Description{Our method wins, especially in the empty map}
\end{figure}

In Figure~\ref{fig:results}, for each value of $N=1,\ldots,100$ on the $x$-axis, the curve labeled \textbf{Ours} and the curve labeled \textbf{PPvP} represent the cost of the best plan obtained from optimizing the first $N$ initial plans, averaged over $100$ instances. The line labeled \textbf{OO} represents the cost of optimization of the optimal plan on the grid, also averaged over $100$ instances.
Figure~\ref{fig:wr_graph} shows the winning rates. For each method, the curve represents the percentage of instances where the best plan among the first $N$ plans of all methods is generated by it, where $N=1,\ldots,100$.
If costs are the same, the instance is counted for both methods.

In the empty map, \textbf{OO} tended to produce the best trajectories when only a few initial solutions were considered.
Notably, \textbf{OO} was not the best even at $N=1$ in the map with obstacles. This may be because reducing costs by making clever moves on narrow parts of the grid would make the optimized trajectories worse.

\begin{figure}[t]
    \includegraphics[width=\columnwidth]{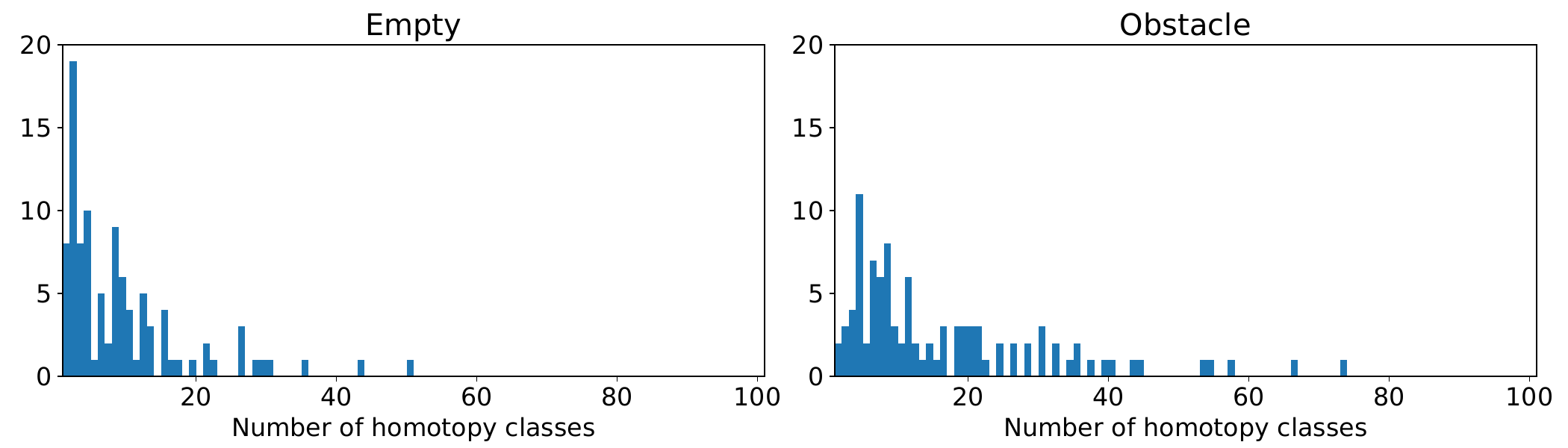}
    \caption{Histogram of the number of homotopy classes present in $100$ solutions generated with \textbf{PPvP}, for $100$ instances.}
    \label{fig:histogram}
    \Description{In most instances, numbers are less than 20, especially in empty map}
\end{figure}

For both maps, \textbf{Ours} tended to produce better solutions than \textbf{PPvP} in terms of both average cost and winning rate.
The difference was larger for the empty map than for the map with obstacles.

To confirm that homotopical diversity of initial plans contributed to the minimum cost, 
we counted the number of homotopically distinct solutions generated with \textbf{PPvP} for each instance. The distribution is shown in Figure~\ref{fig:histogram} as histograms.
The results indicate that \textbf{PPvP} generated few homotopies, particularly on the empty map. This is expected given the lack of homotopy divergence due to obstacles on the empty map.
The difference in the homotopical diversity of the plans generated by \textbf{PPvP} between maps is thought to be responsible for the difference in results after optimization.
Thus, homotopical diversity is actually important for avoiding local optima and obtaining global optima.
Our method is particularly effective on maps with few obstacles, where it is difficult to generate homotopically distinct plans naively.

\section{Conclusion and Future Work}\label{sec:conclusion}
\begin{figure}[t]
    \centering
    \includegraphics[width=0.20\columnwidth]{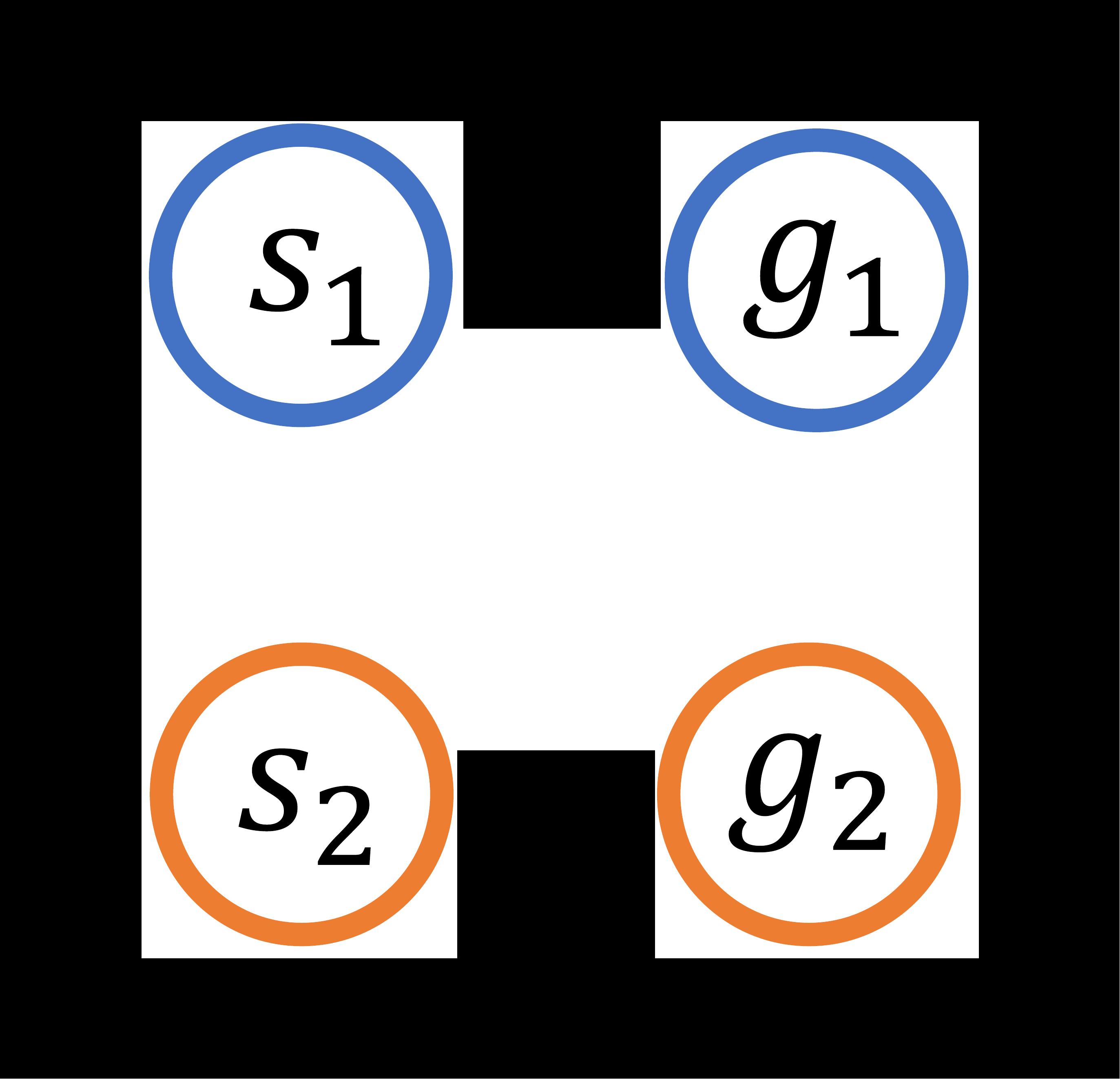}
    \caption{Example scenarios in which homotopy depends on order in which agents pass through the narrow part.}
    \label{fig:narrow}
    \Description{Field with a narrow part and two pairs of starts and goals. Starts are on the opposite side of goals across the narrow part.}
\end{figure}

We proposed a practical method for multi-agent path planning in planar domains with obstacles. We used Dynnikov coordinates and revised prioritized planning to efficiently generate homotopically distinct plans. Our experimental results indicate that our method is significantly faster than the method using the Dehornoy order.
We also conducted an experiment to optimize the generated plans and compared them with those generated with baseline methods. The results indicated that generating homotopically distinct plans using our method led to an improvement in optimized plans. This showcases the effectiveness of our approach in achieving lower-cost trajectories for multi-agent path planning.

One of the main limitations of our approach is that it cannot consider the homotopical differences that arise only when both sizes of agents and obstacles are taken into account. In scenarios where several agents need to pass through a narrow part where only one agent can pass at a time, homotopy classes of solutions can split, depending on the order in which the agents pass through the narrow section. Figure~\ref{fig:narrow} illustrates such a situation, where the homotopy classes of solutions depend on whether the first agent goes through the narrow section earlier or the second agent does.
This phenomenon does not occur when the size of agents can be ignored, as assumed with our approach. There is a computational-geometric difficulty when computing these differences. 

Another limitation is the abandonment of optimality. While simple A* search is workable, it becomes impractical for larger problem instances~\citep{stern2019multi2}. To improve efficiency without sacrificing optimality, it may be possible to combine our framework with efficient optimal approaches.
To calculate braids while searching, it is required to determine the moves of all agents.
This makes it challenging to use conflict-based search, where agent paths are planned independently in its low-level search.
On the other hand, increasing-cost tree search could be a more suitable candidate since it takes into account combinations of all agent paths in its low-level search.
Another non-trivial challenge is to make reduction-based methods homotopy-aware.

A potential extension of our research involves exploring decentralized control strategies.
Specifically, it may be possible to achieve effective coordination among agents with minimal communication by conveying only braids to other agents. 

\bibliographystyle{plainnat}
\bibliography{ref.bib}

\appendix

\section{On Presentations of Pure Braid Groups}\label{appendix:pn}
\citet{bhattacharya2018path} gave a presentation of a homotopy group for path planning with $n$ agents on a plane, which is a pure braid group $P_n$, with the following generators:
\begin{equation}
\left\{u_{i,j/\gamma_{i+1},\ldots,\gamma_{j-1}}\middle|1\leq i<j\leq n,\,\gamma_{i+1},\ldots,\gamma_{j-1}\in\{+,-\}\right\}.
\end{equation}

The relations are as follows.
\begin{itemize}
    \item For $i<j<k$, $\alpha_{i+1},\ldots,\alpha_{j-1}\in\{+,-\}$, and $\beta_{j+1},\ldots,\beta_{k-1}\in\{+,-\}$,
    \begin{equation}
    \begin{split}
    &u_{i,j/\alpha_{i+1},\ldots,\alpha_{j-1}}
    \cdot
        u_{i,k/\alpha_{i+1},\ldots,\alpha_{j-1},-,\beta_{j+1},\ldots,\beta_{k-1}}
        \cdot u_{j,k/\beta_{j+1},\ldots,\beta_{k-1}}\\
        &\cdot u^{-1}_{i,j/\alpha_{i+1},\ldots,\alpha_{j-1}}
    \cdot
        u^{-1}_{i,k/\alpha_{i+1},\ldots,\alpha_{j-1},+,\beta_{j+1},\ldots,\beta_{k-1}} 
        \cdot u^{-1}_{j,k/\beta_{j+1},\ldots,\beta_{k-1}},
    \end{split}
    \end{equation}
    \begin{equation}
    \begin{split}
    &u_{i,j/\alpha_{i+1},\ldots,\alpha_{j-1}}
    \cdot
        u_{i,k/\alpha_{i+1},\ldots,\alpha_{j-1},-,\beta_{j+1},\ldots,\beta_{k-1}}\\
        &\cdot u^{-1}_{i,j/\alpha_{i+1},\ldots,\alpha_{j-1}}\cdot
        u^{-1}_{i,k/\alpha_{i+1},\ldots,\alpha_{j-1},+,\beta_{j+1},\ldots,\beta_{k-1}},
    \end{split}
    \end{equation}
    and
    \begin{equation}\label{eq:p-rel}
    \begin{split}
    &u_{i,k/\alpha_{i+1},\ldots,\alpha_{j-1},-,\beta_{j+1},\ldots,\beta_{k-1}}
        \cdot u_{j,k/\beta_{j+1},\ldots,\beta_{k-1}}\\
    &\cdot u^{-1}_{i,k/\alpha_{i+1},\ldots,\alpha_{j-1},+,\beta_{j+1},\ldots,\beta_{k-1}} 
    \cdot u^{-1}_{j,k/\beta_{j+1},\ldots,\beta_{k-1}}.
    \end{split}
    \end{equation}
\item Let $i,j,i',j'$ be distinct indices with $i<j,\,i'<j',\,i<i'$. Let $\gamma_{i+1},\ldots,\gamma_{j-1}\in\{+,-\}$ and $\gamma'_{i'+1},\ldots,\gamma'_{j'-1}\in\{+,-\}$ be signs.
When $i<i'<j<j'$, we assume that $\gamma_{i'}\neq \gamma'_j$ and $(\gamma_k,\gamma'_k)\neq(-\gamma_{i'},\gamma_{i'})$ for all $i'<k<j$.
When $i<i'<j'<j$, we assume that $\gamma_{i'}= \gamma_{j'}$ and $(\gamma_k,\gamma'_k)\neq(-\gamma_{i'},\gamma_{i'})$ for all $i'<k<j'$. For such tuples,
    \begin{equation}\label{eq:commute}
    \begin{split}
    &u_{i,j/\gamma_{i+1},\ldots,\gamma_{j-1}}
    \cdot u_{i',j'/\gamma'_{i'+1},\ldots,\gamma'_{j'-1}}
    \cdot u^{-1}_{i,j/\gamma_{i+1},\ldots,\gamma_{j-1}}
    \cdot u^{-1}_{i',j'/\gamma'_{i'+1},\ldots,\gamma'_{j'-1}}.
    \end{split}
    \end{equation}
\end{itemize}
The description in the original paper is incomplete as it omits the condition for the relation (\ref{eq:commute}).

We consider the word
\begin{equation}
w=u_{1,4/--}u_{2,4/-}u_{3,4}u_{1,4/++}^{-1}u_{2,4/+}^{-1}u_{3,4}^{-1},
\end{equation}
which is irreducible when using Dehn's algorithm.
On the other hand,
\begin{equation}
\begin{split}
w &=u_{2,4/-}u_{1,4/+-}u_{3,4}u_{1,4/++}^{-1}u_{2,4/+}^{-1}u_{3,4}^{-1}\\
&=u_{2,4/-}u_{3,4}u_{2,4/+}^{-1}u_{3,4}^{-1}\\
&=u_{3,4}u_{3,4}^{-1}\\
&=e,
\end{split}
\end{equation}
where the first, second, and third equalities are deduced from (\ref{eq:p-rel}) with $(i,j,k)=(1,2,4)$, $(1,3,4)$, and $(2,3,4)$, respectively. Thus, Dehn's algorithm is incomplete for this presentation when $n\geq 4$.

The pure braid group $P_n$ has a standard presentation~\citep{rolfsen2010tutorial} using generators $\left\{a_{i,j}\middle|1\leq i<j\leq n\right\}$ with notation (\ref{eq:def_a})
and the following relations.
\begin{itemize}
\item For $1\leq i<j<k\leq n$,
\begin{equation}
    a_{i,j}a_{i,k}a_{j,k}a_{i,j}^{-1}a_{j,k}^{-1}a_{i,k}^{-1},
\end{equation}
and
\begin{equation}
    a_{i,k}a_{j,k}a_{i,j}a_{i,k}^{-1}a_{i,j}^{-1}a_{j,k}^{-1}.
\end{equation}
\item For $1\leq i<j<k<l\leq n$,
\begin{equation}
    a_{i,j}a_{k,l}a_{i,j}^{-1}a_{k,l}^{-1},
\end{equation}
\begin{equation}
    a_{i,l}a_{j,k}a_{i,l}^{-1}a_{j,k}^{-1},
\end{equation}
and
\begin{equation}
    a_{i,k}a_{j,k}a_{j,l}a_{j,k}^{-1}a_{i,k}^{-1}a_{j,k}a_{j,l}^{-1}a_{j,k}^{-1}.
\end{equation}
\end{itemize}
Let
\begin{equation}
u_{i,j/\sigma_{i+1},\ldots,\sigma_{j-1}}=a_{i,i+1}^{d_{i+1}}\cdots a_{i,j-1}^{d_{j-1}} a_{i,j}a_{i,j-1}^{-d_{j-1}}\cdots a_{i,i+1}^{-d_{i+1}},
\end{equation}
where $d_k=1$ if $\sigma_k=+$, and $d_k=0$ if $\sigma_k=-$.
Then, the relations for $\{u_{i,j/\sigma_{i+1},\ldots,\sigma_{j-1}}\}$ can be deduced from the relations for $\{a_{i,j}\}$ and vice versa. Therefore, we can translate words from Bhattacharya and Ghrist's presentation to words for the standard presentation. However, this translation increases the lengths of words by a factor of $O(n)$.
On the other hand, lengths of words constructed by Bhattacharya and Ghrist's method are equal to those of words constructed by the method in \S~\ref{subsec:word construction}.
To the best of our knowledge, there is no algorithm for the word problem of the pure braid group that is efficient enough to compensate for these drawbacks. This is why we use elements of the braid group to label homotopy classes, instead of those of the pure braid group.

\section{Proof of Proposition~\ref{prop:homotopy_inj}}\label{appendix:proof}
The map from $\mathcal{C}_{r+n}(\mathbb{R}^2)$ to $\mathcal{C}_r(\mathbb{R}^2)$, which sends $(p_1,\ldots,p_{r+n})$ to $(p_1,\ldots,p_r)$, is a fiber bundle~\citep{fadell1962configuration}, and the fiber at $(o_1,\ldots, o_r)$ is the image of the embedding (\ref{eq:embedding}). Thus, the following exact sequence of homotopy groups is induced:
\begin{equation}        
\pi_2\left(\mathcal{C}_r(\mathbb{R}^2)\right)\to
\pi_1\left(\mathcal{C}_n(\mathbb{R}^2\setminus\{o_1,\ldots,o_r\})\right)\to \pi_1\left(\mathcal{C}_{r+n}(\mathbb{R}^2)\right),
\end{equation}
where $\pi_2\left(\mathcal{C}_r(\mathbb{R}^2)\right)$ is the second homotopy group of $\mathcal{C}_r(\mathbb{R}^2)$~\citep{hatcher2002algebraic}.
Moreover, $\pi_2\left(\mathcal{C}_r(\mathbb{R}^2)\right)$ is trivial~\citep{knudsen2018configuration}.

\end{document}